\documentclass[10pt,a4paper]{amsart}
\usepackage[cp1250]{inputenc}
\usepackage[english]{babel}
\usepackage{color}
\usepackage[colorlinks,linkcolor=blue,citecolor=red]{hyperref}
\usepackage{amssymb}
\usepackage{amsmath}
\usepackage{amsfonts}
\usepackage[IL2]{fontenc}
\usepackage{url}
\usepackage{graphicx}

\usepackage{pgf,tikz}
\usepackage{mathrsfs}
\usetikzlibrary{arrows}
\usetikzlibrary{snakes}
\usetikzlibrary{arrows}
\usetikzlibrary{calc}
\usetikzlibrary{patterns}
\usepackage{tkz-fct}

\newtheorem{defi}{\bf D\scriptsize EFINITION \normalsize}
\newtheorem{theorem}{\bf T\scriptsize HEOREM \normalsize}
\newtheorem{lm}{\bf L\scriptsize EMMA \normalsize}
\newtheorem{dk}{\bf C\scriptsize OROLLARY \normalsize}
\newtheorem{rem}{\bf R\scriptsize EMARK \normalsize}
\newtheorem{exa}{\bf E\scriptsize XAMPLE \normalsize}
\newtheorem{pro}{\bf P\scriptsize ROBLEM \normalsize}
\newtheorem{prop}{\bf P\scriptsize ROPOSITION \normalsize}
\newtheorem{no}{\bf N\scriptsize OTE \normalsize}

\newenvironment{remark}{\begin{rem}\rm}{\end{rem}}
\def\kopr{\hfill\raisebox{3pt}{\framebox{$\star$}}}
\newenvironment{example}{\begin{exa}\rm}{$\kopr$\end{exa}}
\newenvironment{definition}{\begin{defi}\rm}{\end{defi}}

\newenvironment{corollary}{\begin{dk}\it}{\end{dk}}
\newenvironment{proposition}{\begin{prop}\it}{\end{prop}}

\newcommand{\inte}[2]{\int \limits_{#1}^{#2}}

\newcommand{\nadsebou}[2]{\begin{array}{c} #1 \\ #2 \end{array}}

\newcommand{\zav}[1]{\left( #1 \right)}

\newcommand{\imag}{{\rm i}}

\newcommand{\abs}[1]{\left| #1 \right|}

\allowdisplaybreaks
\begin{document}
\title{Pedal coordinates, Dark Kepler and other force problems}
 \author{Petr Blaschke}
 \thanks{Supported by GA \v CR grant no. 201/12/G028}
\address{ Mathematical Institute, Silesian University in Opava, Na Rybnicku 1, 746 01 Opava, Czech Republic}

\email{Petr.Blaschke@math.slu.cz}
\begin{abstract} 
We will make the case that \textit{pedal coordinates} (instead of polar or Cartesian coordinates) are more natural settings in which to study force problems of classical mechanics in the plane. We will show that the trajectory of a test particle under the influence of central and Lorentz-like forces can be translated into pedal coordinates at once without the need of solving any differential equation. This will allow us to generalize Newton theorem of revolving orbits to include nonlocal transforms of curves. Finally, we apply developed methods to solve the ``dark Kepler problem'', i.e. central force problem where in addition to the central body, gravitational influences of dark matter and dark energy are assumed.
\end{abstract}
\maketitle
\section{Introduction}\label{Intro}
Since the time of Isaac Newton it is known that conic sections offers full description of trajectories for the so-called Kepler problem -- i.e. central force problem, where the force varies inversely as a square of the distance:
$
F\propto \frac{1}{r^2 }.
$

There is also another force problem for which the trajectories are fully described, Hook's law, where the force varies in proportion with the distance:
$
F\propto r.
$
(This law is usually used in the context of material science but can be also interpreted as a problem of celestial mechanics since such a force would produce gravity in a spherically symmetric, homogeneous bulk of dark matter by Newton shell theorem.)

Solutions of Hook's law are also conic sections 
but with the distinction that the origin is now in the center (instead of the focus) of the conic section.

But save for the law of gravity and Hook's law there seems to be no other force problem whose trajectories are fully described in known curves. Indeed, by Bertrand's theorem \cite{Bertrand} such a description would be problematic at best, since it states that \emph{no other} central force problem (i.e. when the force is only a function of distance from the center and points to it) has the property that all bounded curves are also closed.

But Newton himself (in \emph{Philosophi\ae\ Naturalis Principia Mathematica}) proves that this effort is not hopeless after all by showing that there is another force inversely proportional to the cube of the distance:
$
F\propto \frac{1}{r^3},
$
whose contribution can be understood purely geometrically. More precisely, he observed, that if a curve is given as a solution to the central force problem $F(r)$ adding additional force of the form
$
\frac{L^2}{m r^3}(1-k^2),
$
where $L$ is the particle's angular momentum and $m$ its mass, is equivalent to a making the curve's $k$-th harmonic.

The $k$-th harmonic of a curve is done simply by multiplying the angle of every point on a curve by a constant $k$. For example, the following picture shows the second harmonic ($k=2$) and the third subharmonic ($k=\frac{1}{3}$) of an ellipse, where the center of polar coordinates are in one focus:
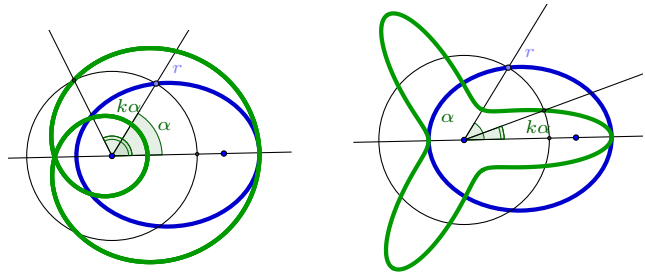
\begin{center}
\definecolor{qqzzqq}{rgb}{0.,0.6,0.}
\definecolor{qqwuqq}{rgb}{0.,0.39215686274509803,0.}
\definecolor{uuuuuu}{rgb}{0.26666666666666666,0.26666666666666666,0.26666666666666666}
\definecolor{xdxdff}{rgb}{0.49019607843137253,0.49019607843137253,1.}
\definecolor{qqqqcc}{rgb}{0.,0.,0.8}
\definecolor{qqqqff}{rgb}{0.,0.,1.}

\end{center}

Notice that these curves are closed but this does not contradict Bertrand's theorem since not \emph{all} harmonics are closed (only those with rational $k$).

This theorem (of revolving orbits) remains largely forgotten until 1997, when it was studied in works \cite{Bell1},\cite{Bell2}. A generalization was discovered by Mahomed and Vawda in 2000 \cite{Mahomed}. They assumed that the radial distance $r$, and the angle $\varphi$ changes according to rule:
\begin{equation}\label{Mahomedtr}
r\to \frac{ar}{1-br},\qquad \varphi \to \frac{1}{k}\varphi,
\end{equation}
where $a,b$ are given constant.

They proved that such a transform of the solution is equivalent to changing the force as follows:
\begin{equation}\label{Mohamedeq}
F(r)\to \frac{a^3}{(1-br)^2}F\zav{\frac{ar}{1-br}}+\frac{L^2}{m r^3}(1-k^2)-\frac{bL^2}{mr^2},
\end{equation}
where again $m$ is the particle mass and $L$ its angular momentum.

This result is complete as far as ``point'' transformations are considered. But -- as we will prove -- it can be generalized for a quite large class of \emph{nonlocal} transforms. More precisely:
\begin{theorem}\label{nonlocalrevolvingorbits}
Consider a transform $T_f$ depending on a smooth function $f$ that maps a planar curve given in polar coordinates $h(r,\varphi)=0$ into a curve in polar coordinates $h(\tilde r,\tilde\varphi)=0$ according to the rule:
$$
\tilde r=f(r),\qquad \tilde \varphi=\frac{1}{k}\inte{\varphi_0}{ \varphi}\frac{f^\prime( r(t)) r(t)^2}{f^2( r(t))}{\rm d}t.
$$

Let a plane curve $\gamma$ be a solution to the central force problem
$$
\ddot x=F(r)\frac{x}{r},\qquad x\cdot \dot x^\perp=L,
$$
where $L$ is the angular momentum, $\dot x^\perp$ is the vector perpendicular to $\dot x$ and $r:=\abs{x}$.

Then the transformed curve $T_f (\gamma)$ is a solution to the central force problem:
$$
\ddot x=\zav{f^\prime F(f)-\frac{L^2 k^2}{r^3}+\frac{L^2 f^\prime}{f^3}}\frac{x}{r},\qquad x\cdot \dot x^\perp=Lk,
$$
where $f\equiv f(r)$. 
\end{theorem}
Notice that previous two results are included in this theorem as a special cases for $f(r):= r$ (Newton) and $f(r):= \frac{ar}{1-br}$ (\cite{Mahomed})\footnote{To get an exact match between Theorem \ref{nonlocalrevolvingorbits} and equation (\ref{Mohamedeq}) a little bit of scaling is needed. Concretely, in equation (\ref{Mohamedeq}) set $m:=1$ and divide the right hand side by $a^2$.}. In fact, those are the only instances of $f$, where the transform $T_f$ is local, i.e. the integrand is independent of $r$, that is 
$$
\frac{f^\prime r^2}{f^2}=1.
$$

Nonlocal transforms are quite natural to assume. 
Take for example the transform $A_\omega$ which rotate each point on a given curve around some fixed point by amount which is proportional to the area swept by the radius vector from some initial angle $\varphi_0$, i.e.
$$
\tilde r= r,\qquad \tilde \varphi= \varphi- \varphi_0+\omega\inte{ \varphi_0}{ \varphi} r^2(t){\rm d}t.
$$
The conservation of angular momentum in central force problems tells us exactly that such an area is proportional to the time it took to the test particle to travel from angle $\varphi_0$ to $\varphi$. Thus the transform $A_\omega$ describes how curves, given by a central force problem, change when passing to a rotating frame of reference (which is done on regular basis in celestial mechanics).


The main goal of this article is to make the case that results similar to Theorem \ref{nonlocalrevolvingorbits} and problem of classification of orbits are best studied in the so-called \emph{pedal coordinates}.

Pedal coordinates (\cite{Yates,Edwards}) describe the position of a point $x$ on a given curve $\gamma$ by two numbers $(r,p)$ -- where $r$ is the distance of $x$ from the origin and $p$ is the distance of origin to the tangent of $\gamma$ at $x$.

In pedal coordinates, we can even assume more general force problems than central ones. Specifically, we can include the situation when the force has Lorentz like (or magnetic) component -- i.e. a component that is pointing in the direction perpendicular to the velocity of a test particle and depending on the distance.

In fact, as we show, solutions of such force problems can be translated into pedal coordinates \emph{algebraically without any integration whatsoever} showing that they are indeed ``natural'' coordinates for the job. More precisely, we prove the following theorem:
\begin{theorem}\label{cfth} Consider a dynamical system:
\begin{equation}\label{dynsys}
\ddot x=F^\prime\zav{\abs{x}^2}x+2 G^\prime\zav{\abs{x}^2}{\dot x}^\perp,
\end{equation}
describing an evolution of a test particle (with position $x$ and velocity $\dot x$) in the plane in the presence of central $F$ and Lorentz like $G$ potential. The quantities:
$$
L=x\cdot \dot x^\perp+G\zav{\abs{x}^2}, \qquad c=\abs{\dot x}^2-F\zav{\abs{x}^2},
$$
are conserved in this system.

Then the curve traced by $x$ is given in pedal coordinates by
\begin{equation}
\frac{\zav{L-G(r^2)}^2}{p^2}=F(r^2)+c,
\end{equation}
with the pedal point at the origin. Furthermore, the curve's image is located in the region given by
\begin{equation}\label{regine}
\frac{\zav{L-G(r^2)}^2}{r^2}\leq F(r^2)+c.
\end{equation}
\end{theorem}

The structure of the paper is the following: In Section \ref{pedalsec} we give an introduction to the Pedal coordinates since -- in the author's opinion -- this subject is largely forgotten and hence we do not assume any background knowledge on the reader's part.

Theorem \ref{cfth} tells us that answering the question: 
\\

\emph{For a given curve $\gamma$, what forces do we have to impose on a test particle to move along it?}
\\

is straightforward if we are provided with pedal equation of $\gamma$.

In Section \ref{MTPCS} we derive an efficient method how to translate quite general curves into pedal coordinates, provided that they are given as a solution of a autonomous differential equation (of any order) in polar coordinates. This will be the statement of Proposition \ref{MPC} -- a result which (to the best of the author's knowledge) is new.

Characterization of orbits in known curves, of course, depends on which curves are known. There are entire books filled with interesting curves (for example \cite{Yates},\cite{Lawrence}) but instead of memorizing them all, it is often better to look for some ``transforms'' that connects them, e.g. Pedal, circle inverse, parallel, dual, etc.


Some of these classical transforms will be introduced in Sections \ref{pedalsec} and \ref{transsec}, along with a method how to translate known transforms in polar coordinates into pedal coordinates (Theorem \ref{transth}).

Along this line in Section \ref{evolutesec} we also provide a pedal analogue to more advanced transforms like evolute, involute, contrapedal and catacaustic -- which the author was also unable to find in the literature (Proposition \ref{evolute}).

In Section \ref{CentralForce} we give a proof of Theorem \ref{cfth} and also (basically as a corollary) proof of Theorem \ref{nonlocalrevolvingorbits}.

Finally, we apply these results to a concrete examples. First, we focus on the relativistic version of the Kepler problem in Section \ref{RKP}, where we have successfully managed to classify solutions in terms of $sn$-spirals, introduced in Section \ref{spiralssec} -- a generalization of the famous sinusoidal spirals.

We also tackle the ``dark Kepler problem'' in Section \ref{DKP} (in addition to the central body, gravitational influences of dark matter and dark energy are allowed), where we show that a particular solution is the Cartesian oval, as seen from a rotating frame of reference. 
Furthermore we show that additional solutions can be obtained using a nonlocal transform that can be constructed from the rotating frame transforms $A_\omega$ and point transforms. 

The author would like to thank Miroslav Engli\v s and Michal Marvan for careful reading of the manuscript and suggesting numerous improvements.
\section{Pedal coordinates}\label{pedalsec}

Remember that the ``pedal coordinates'' of a point $x$ on a differentiable curve $\gamma$ in the plane are given by two positive real numbers $(r,p)$, where $r$ is the distance of $x$ from some given point $O$ (the so called \emph{pedal point}) and $p$ is the distance of $O$ from the tangent line of $\gamma$ at $x$.
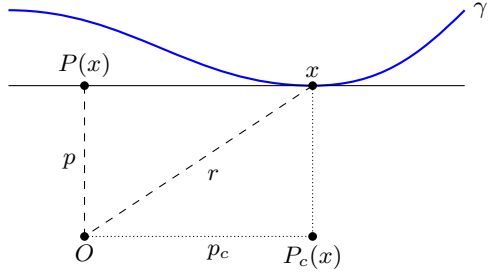
\begin{center}
{\small
\begin{tikzpicture}[domain=-4:5]
\draw [thick,color=blue] (-4,3) to [out=0,in=180] (0,2) to [out=0, in=225] (2,3);
\coordinate [label=right:{$\gamma$}] (M1) at (2,3);
\coordinate [label=above:{$x$}] (M2) at (0,2);
\draw [thin] (-4,2) to (2,2);
\draw [dashed] (-3,2) to (-3,0);
\draw [dashed] (0,2) to (-3,0);
\filldraw (-3,0) circle (0.05);
\filldraw (0,2) circle (0.05);
\filldraw (-3,2) circle (0.05);
\coordinate [label=below:{$O$}] (M3) at (-3,0);
\coordinate [label=above:{$P(x)$}] (M4) at (-3,2);
\draw [densely dotted] (0,2) to (0,0);
\draw [densely dotted] (-3,0) to (0,0);
\filldraw (0,0) circle (0.05);
\coordinate [label=below:{$P_c(x)$}] (M5) at (0,0);
\coordinate [label=left:{$p$}] (M6) at (-3,1);
\coordinate [label=below right:{$r$}] (M7) at (-1.5,1);
\coordinate [label=below right:{$p_c$}] (M8) at (-1.5,0);
\end{tikzpicture}
}
\end{center}

It is useful to measure also the distance of $O$ to the normal (the ``contrapedal coordinate'' $p_c$) even though it is not an independent quantity and it relates to $(r,p)$ as 
$$
p_c:=\sqrt{r^2-p^2}.
$$
For every point $x\in \gamma$ we can define two additional points denoted in the picture $P(x), P_c(x)$ and thus create two additional curves. Let us denote by $P(\gamma)$ the \emph{pedal curve} -- i.e. the locus of points $P(x)$. And by $P_c(\gamma)$ the \emph{contrapedal curve} -- i.e. the locus of points $P_c(x)$.

In fact, one of the main advantages of pedal coordinates is that the operation of making pedal curve (which would in general require solving a differential equation in Cartesian coordinates) can by done by simple algebraic manipulation.

Concretely (see \cite[p. 228]{williamson}), to any curve $\gamma$ given by the equation
$$
f(p,r)=0,
$$ 
in pedal coordinates, the pedal curve $P(\gamma)$ satisfies the equation
$$
f\zav{r,\frac{r^2}{p}}=0.
$$

The contrapedal curve $P_c(\gamma)$ is much harder to obtain, but for specific examples it can be done as well (as we will see).

Even making the curve's harmonics, introduced in Section \ref{Intro} is easily done in pedal coordinates:
\begin{eqnarray}\label{harmtrans}
f\zav{\frac{1}{p^2},r}=0&\stackrel{H_k}{\longrightarrow} & f\zav{\frac{k^2}{p^2}-\frac{k^2-1}{r^2},r}=0.
\end{eqnarray}

Many additional ``transforms'' (i.e. operations that bring curves into different curves) can be described easily in pedal coordinates, for example:
\begin{eqnarray}
\label{translist}f(p,r)=0 &\stackrel{S_\alpha}{\longrightarrow}& f(\alpha p, \alpha r,)=0,\\
f(p,r)=0 &\stackrel{I_R}{\longrightarrow}& f\zav{\frac{R p}{r^2},\frac{R}{r}}=0,\\
f\zav{p,r^2}=0 &\stackrel{E_c}{\longrightarrow}& f\zav{p-c,r^2-2pc+c^2}=0,\\
f(p,r)=0 &\stackrel{D_R=I_R P}{\longrightarrow}& f\zav{\frac{R }{r},\frac{R}{p}}=0,\\
f\zav{\frac{1}{p^2},r}=0 &\stackrel{F_c= P E_c P^{-1}}{\longrightarrow}& f\zav{\frac{1}{(r-c)^4}\zav{\frac{r^4}{p^2}-2cr+c^2},r-c}=0,\\
\label{translistp} f\zav{\frac{1}{p^2},r}=0 &\stackrel{E^\star_c:=D_1 E_c D_1 =I_1 F_c I_1}{\longrightarrow}& f\zav{\frac{1}{p^2}-\frac{2c}{r}+c^2,\frac{r}{1-cr}}=0,
\end{eqnarray}
where:

$S_\alpha$ is the scaling of a curve by a factor $\alpha$.

$I_R$ is the circle inverse with respect to a circle at pedal point with radius $R$.

$D_R$ is the dual curve, i.e. a curve in the dual projective space consisting of the set of lines tangent to the original curve.

$E_c$ is the parallel (or equidistant) curve at distance $c$.

$F_c$ maps every point $x\to x+c\frac{x}{\abs{x}}$, i.e. shifts a curve away from the pedal point $O$ by the fixed amount $c$ (which can be negative).

And finally, the $E^\star_c$ changes the radial component of a curve by $r\to \frac{r}{1-cr}$.

These transformation formulas are not hard to derive as we will see in Section \ref{transsec}. 
\subsection{Selected curves}
Simplest curves translates into pedal coordinates as follows:
\subsubsection{Line} The pedal equation of a line is obviously
$$
p=a,
$$
where $a\geq 0$ is the distance of the line from the pedal point $O$.

It is important to note that the other coordinate $r$ is \emph{not completely} arbitrary. In addition to being non-negative, it has to always be true that
$$
p\leq r,
$$
in other words the defining quality of the Euclidean space that ``shortest distance is the straight line'' must be obeyed. This puts on radius $r$ the constrain
$$
r\geq a.
$$ 
\subsubsection{Point}
The complement of a line is the curve given by
$$
r=a,\qquad a\geq0,
$$
which is \emph{not} a circle but it is actually a point distant $a$ from $O$. ($r=a$ is a circle in polar coordinates since the other coordinate -- the angle $\varphi$ -- is arbitrary. In pedal coordinates the other coordinate satisfies $p\leq a$ -- which is consistent with the picture that a curve consisted of single point has an arbitrary tangent line at this point.)
\subsubsection{Circle}
Combining these two equations into one
$$
p=R,\qquad r=R,
$$
we have finally arrived at the pedal equation of a circle (with center at the pedal point and radius $R$).

Unlike with Cartesian coordinates, where the change of origin is done trivially by simple translation, the position of pedal point influences heavily the form of pedal equation and there is no simple formula telling us how a change of pedal point changes the pedal equation.

For example, the pedal equation of a circle with pedal point on the circle is
$$
2Rp=r^2,
$$
and for arbitrary pedal point it looks as
$$
2Rp=r^2+R^2-a^2,
$$
where $a\geq 0$ is the distance of the center of circle to pedal point. The same equation can be described using contrapedal coordinate as
$$
p_c^2+(p-R)^2=a^2.
$$
(These equation are not obvious, but they can be neatly derived, once we understand how curvature translates into pedal coordinates.)

We can also consider the equation
$$
p=r,
$$
which describes \emph{all} the concentric circles with center at pedal point.
\begin{remark}
At this point it should be clear that pedal coordinates does not tell us everything about the curve and they actually describes many curves at once -- if you choose to differentiate between them.

The equation $p=a$ is valid for \emph{any} line distant $a$ and $r=a$ for any point distant $a$, etc.

Obviously, the pedal coordinates do not care about rotation around the pedal point and about the curve's parametrization, but it is actually not easy to tell in general the nature of ambiguity associated to a pedal equation -- in fact, it differs from equation to equation. As we will see in the next section, what we can tell is that all curves that a pedal equation describes are solutions to a nonlinear first order autonomous differential equation in polar coordinates.

This is actually \emph{an advantage} of pedal coordinates over other systems if you are interested only in the general shape of the curve and do not want to be distracted by details. 
\end{remark}
\subsubsection{Logarithmic spiral}
Next curve which satisfies linear equation in pedal coordinates
$$
p=a r,\qquad 0\leq a\leq 1,
$$
is the logarithmic spiral, where $a=\abs{\sin \alpha}$ and $\alpha$ is the angle between tangent and radial line, which is constant for logarithmic spirals.
\subsubsection{Circle involute and Spiral of Archimedes}
The curve which is sort of a contrapedal version of a line, i.e. satisfies
$$
p_c=a, \qquad a\geq 0,
$$
can be shown to be the ``Involute of a circle'' (with pedal point at the center).\footnote{It is a shame that no better name is given to such an important curve. It is as if we would say instead of ``parabola'' just ``antipedal of a line''.} This curve is often mistaken for Spiral of Archimedes, which is actually its pedal, i.e.
$$
p_c=a\qquad \stackrel{P}{\longrightarrow} \qquad \frac{1}{p^2}=\frac{1}{r^2}+\frac{c^2}{r^4}.
$$
The difference between those two curves can be understood as follows: While the legs of Circle involute are parallel:
$$
p_c=a\qquad \stackrel{E_c}{\longrightarrow} \qquad p_c=a,
$$ 
the legs of Archimedian spiral are spread constantly in a radial way, i.e.
$$
\frac{1}{p^2}=\frac{1}{r^2}+\frac{c^2}{r^4}\qquad \stackrel{F_c}{\longrightarrow} \qquad \frac{1}{p^2}=\frac{1}{r^2}+\frac{c^2}{r^4},
$$ 
which is easy to see since $F_c=P E_c P^{-1}$.
\section{Moving to pedal coordinates}\label{MTPCS}
It is known that a curve given in polar coordinates $f(\varphi,r)=0$ can be expressed in pedal coordinates by eliminating $\varphi$ from the equations
$$
f(\varphi,r)=0,\qquad p=\frac{r^2}{\sqrt{r^2+{r_\varphi^\prime}^2}}.
$$

In terms of the factor $r_\varphi^\prime$, the second equation gives us
$$
\abs{r_\varphi^\prime}=\frac{r p_c}{p}=P(p_c).
$$
Consequently, if a curve can be written as a solution of differential equation
$$
f\zav{r,\abs{r_\varphi^\prime}}=0,
$$
its pedal equation becomes simply
$$
f\zav{r,\frac{r}{p}p_c}=0.
$$
Or, alternatively, we can say that this curve is the pedal of
$$
f(p,p_c)=0.
$$
\begin{example}
As an example take the logarithmic spiral with the spiral angle $\alpha$:
$$
r=a e^{\frac{\cos\alpha}{\sin\alpha} \varphi}.
$$
Differentiating with respect to $\varphi$ we obtain
$$
r_\varphi^\prime= \frac{\cos\alpha}{\sin\alpha} ae^{\frac{\cos\alpha}{\sin\alpha} \varphi}=\frac{\cos\alpha}{\sin\alpha} r,
$$
hence 
$$
\abs{r_\varphi^\prime}=\abs{\frac{\cos\alpha}{\sin\alpha}} r,
$$
and thus in pedal coordinates we get
$$
\frac{r}{p}p_c=\abs{\frac{\cos\alpha}{\sin\alpha}} r \qquad \Rightarrow \qquad
\abs{\sin\alpha} p_c=\abs{\cos\alpha} p,
$$
or using the fact that $p_c^2=r^2-p^2$ we obtain
$$
p=\abs{\sin\alpha}r,
$$
as claimed in the previous section.
\end{example}
\begin{example}
Similarly, Spiral of Archimedes given by
$$
r=a\varphi,\qquad a\geq 0,
$$
can be written as a differential equation
$$
\abs{r_\varphi^\prime}=a,
$$
hence we get
$$
P\zav{p_c=a}\qquad \Rightarrow \qquad 
P\zav{r^2=p^2+a^2}\qquad \Rightarrow \qquad
\frac{r^4}{p^2}=r^2+a^2 \qquad \Rightarrow \qquad
\frac{1}{p^2}=\frac{1}{r^2}+\frac{a^2}{r^4},
$$
as claimed. Notice that first equality above tells right away that Spiral of Archimedes is the pedal of $p_c=a$ which we will see shortly is indeed Involute of a circle.
\end{example}

This approach can be generalized as follows:
\begin{proposition}\label{MPC}
A curve $\gamma$ which a solution of a $n$-th order autonomous differential equation ($n\geq 1$)
$$
f\zav{r,\abs{r_\varphi^\prime},r_\varphi^{\prime\prime},\abs{r_\varphi^{\prime\prime\prime}}\dots,r_\varphi^{2j},\abs{r_\varphi^{(2j+1)}},\dots, r_\varphi^{(n)}}=0,
$$
is the pedal of a curve given in pedal coordinates by
$$
f(p,p_c, p_c p_c^\prime,\zav{p_c p_c^\prime}^\prime p_c,\dots, (p_c\partial_p)^n p)=0.
$$
In other words
\begin{align*}
r&=P(p), & \abs{r_\varphi^\prime}&=P(p_c),\\
r_\varphi^{\prime\prime}&=P(p_c p_c^\prime), & \abs{r_\varphi^{\prime\prime\prime}}&=P\zav{(p_cp_c^\prime)^\prime p_c},\\
&\vdots & &\vdots \\
r_\varphi^{(2j)}&=P\zav{(p_c\partial_p)^{2j}p} &\abs{r_\varphi^{(2j+1)}}&=P\zav{(p_c\partial_p)^{2j+1}p}, & \forall j\in\mathbb{N}_0. 
\end{align*}
\end{proposition}
\begin{proof}
Since we know that
$$
{r_\varphi^\prime}^2=\frac{r^2p_c^2}{p^2}=P(p_c^2),
$$
it follows by the properties of pedal operation $P$ and the chain rule that
$$
2P(p_c p_c^\prime)=P\zav{\frac{\partial p_c^2}{\partial p}}=\zav{\frac{\partial P(p_c^2)}{\partial P(p)}}=\frac{\partial \frac{r^2p_c^2}{p^2} }{\partial r}=\frac{\partial {r_\varphi^\prime}^2}{\partial r}=2r_\varphi^{\prime\prime}.
$$
Hence
$$
P\zav{(p_c\partial_p)^{2n} p}=P\zav{\zav{p_cp_c^\prime\partial_p+p_c^2 \partial_p^2}^n p}=\zav{P\zav{p_cp_c^\prime}\partial_r+P(p_c)^2 \partial_r^2}^n r
=\zav{r_\varphi^{\prime\prime}\partial_r+{r_\varphi^\prime}^2\partial_r^2}^n r=\partial_\varphi^{2n} r=r_\varphi^{(2n)},
$$
and
$$
P\zav{(p_c\partial_p)^{2n+1} p}=P\zav{(p_c\partial_p)(p_c\partial_p)^{2n} p}=\frac{rp_c}{p}\partial_r P\zav{(p_c\partial_p)^{2n} p}=\abs{r_\varphi^\prime}\partial_r r^{(2n)}={\rm sgn}(r_\varphi^\prime)r^{(2n+1)}.
$$
Taking the absolute value of the above finishes the proof.
\end{proof}
\begin{corollary}
The radius of curvature $\rho$ and hence the curvature $\kappa$ of a curve $\gamma$ given in pedal coordinates can be computed as
$$
\rho= r r^\prime, \qquad \kappa=\frac{1}{rr^\prime},
$$
where the differentiation is with respect to $p$.
\end{corollary}
\begin{proof}
The signed curvature $\kappa$ is given in polar coordinates by
$$
\kappa:=\frac{{r^2+2{r_\varphi^\prime}^2-rr_\varphi^{\prime\prime}}}{\zav{r^2+{r_\varphi^\prime}^2}^{\frac32}}.
$$
Notice that first derivative $r_\varphi^\prime$ is always raised to the second power hence Proposition \ref{MPC} can be applied and we get
$$
\kappa=P\zav{\frac{{p^2+2p_c^2-p p_c p_c^\prime}}{\zav{p^2+p_c^2}^{\frac32}}}=P\zav{\frac{{2r-pr^\prime}}{r^2}}=\frac{{2\frac{r^2}{p}-r\frac{\partial \frac{r^2}{p}}{\partial r}}}{\frac{r^4}{p^2}}=\frac{{2pr^2r^\prime-2r^2r^\prime p+r^3}}{p^2r^\prime\frac{r^4}{p^2}}=\frac{1}{rr^\prime}.
$$ 
\end{proof}
\begin{example}
We can use the fact that the radius of curvature $\rho$ is a quantity independent on a position of the pedal point, to write the pedal equation for a circle of radius $R$:
$$
\rho=R, \qquad \Rightarrow \qquad rr^\prime=R,
$$
integrating we get
$$
r^2=2Rp+c,
$$
as claimed.
\end{example}
\begin{example}\label{Kp}
Consider as a next example the \emph{Kepler problem} restricted to two dimensional plane, i.e. solution to the system of ODE's:
$$
\ddot{x}=-\frac{M}{\abs{x}^{3}}x,\qquad x\in\mathbb{R}^2,
$$
where $M$ is the reduced mass.
This equation can be written in polar coordinates as
$$
\zav{\frac{1}{r}}_\varphi^{\prime\prime}+\frac{1}{r}=\frac{M}{L^2}, \qquad \dot \varphi=\frac{L}{r^2},
$$
where $L$ is the angular momentum. This can be, of course, easily solved yielding the formula of a conic section. But taking the circle inverse of the first equation, which in polar coordinates can be done simply by
$$
r\to \frac{R}{r}, \qquad \varphi\to -\varphi,
$$
we get
$$
I_R\zav{ r_\varphi^{\prime\prime}+r=R\frac{M}{L^2}},
$$
and moving to pedal coordinates by Proposition \ref{MPC} we arrive at:
$$
I_RP\zav{p_cp_c^\prime+p=\frac{RM}{L^2}}\qquad \Rightarrow \qquad
I_RP\zav{r^2=2\frac{RM}{L^2}p+c}\qquad \Rightarrow \qquad
\frac{R^2}{p^2}=\frac{2M R^2}{L^2 r}+c.
$$ 
This approach gives us not only the pedal equation of the solution but also a way how to construct it. The next to the last expression instruct us that we must take Pedal and then Inverse (in other words Dual) of a circle with radius $\frac{RM}{L^2}$.

It is easy to check that when the pedal point is inside the circle we get an ellipse and for an outside pedal point a hyperbola and for the pedal point on the circle a parabola.

It is easy to see that this circle is in fact the inverse of the solution's circumcircle.
\end{example}
\begin{example}\label{Covalex}
Cartesian oval is defined to be a locus of points for which a linear combination of distances from two foci are constant, i.e.
$$
\abs{x}+\alpha\abs{x-a}=C.
$$
They were studied first by Descartes who realized that such shape can be used to produce lenses without spherical aberration.

In polar coordinates the equation becomes
$$
r+\alpha\sqrt{r^2-2\abs{a}r\cos \varphi+\abs{a}^2}=C.
$$
Solving for $\cos\varphi$ we get
$$
2\abs{a}\alpha^2\cos \varphi=r(\alpha^2-1)-\frac{b^2}{r}+2C,\qquad b^2:=C^2-\alpha^2\abs{a}^2.
$$
The quantity $b^2$ is indeed positive since the origin is inside of the oval, i.e. $\alpha\abs{a}\leq C$. Differentiating with respect to $\varphi$ we obtain
$$
-2\abs{a}\alpha^2\sin\varphi= \zav{\alpha^2-1+\frac{b^2}{ r^2}}r^\prime_\varphi.
$$
Thus we arrive at the differential equation:
$$
4\abs{a}^2\alpha^4=4\abs{a}^2\alpha^4\zav{\cos^2\varphi+\sin^2\varphi}=\zav{r(\alpha^2-1)-\frac{b^2}{r}+2C}^2+\zav{\alpha^2-1+\frac{b^2}{ r^2}}^2{r^\prime_\varphi}^2.
$$
Applying Proposition \ref{MPC} and rearranging we finally obtain:
$$
\frac{\zav{b^2-(1-\alpha^2)r^2 }^2}{4p^2}=\frac{Cb^2}{r}+(1-\alpha^2)C r -\zav{(1-\alpha^2)C^2+b^2}.
$$
It is worth to mentioning that this equation describes two Cartesian ovals (for $\alpha$ and -$\alpha$) simultaneously.
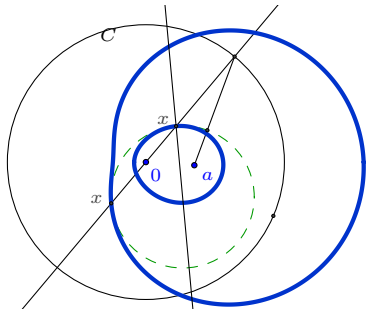
\begin{center}
\definecolor{qqttcc}{rgb}{0.,0.2,0.8}
\definecolor{qqzzqq}{rgb}{0.,0.6,0.}
\definecolor{uuuuuu}{rgb}{0.26666666666666666,0.26666666666666666,0.26666666666666666}
\definecolor{qqqqff}{rgb}{0.,0.,1.}

\end{center} 
\end{example}
\section{Transforms}\label{transsec}
The tools developed in the previous section allow us to translate easily some transforms from polar to pedal coordinates.

Take, for instance, the circle inverse transform $I_R$. As was mentioned earlier, the circle inverse maps a curve in polar coordinates $(r,\varphi)$ into a curve in different polar coordinates $(\tilde r,\tilde \varphi)$ as follows:
$$
\tilde r=\frac{R}{r},\qquad \tilde\varphi =-\varphi.
$$
Hence in pedal coordinates we have
$$
\tilde r=\frac{R}{r},
$$
and
$$
\frac{\tilde r\tilde p_c}{\tilde p}=\abs{\frac{\partial \tilde r}{\partial \tilde\varphi}}=\abs{\frac{\partial \frac{R}{r}}{\partial (-\varphi)}}=\frac{R \abs{r^\prime_\varphi}}{r^2}=\frac{R p_c}{rp}=\frac{\tilde r p_c}{p}.
$$
Thus 
$$
\frac{\tilde p_c}{\tilde p}=\frac{p_c}{p},
$$ 
solving for $\tilde p$ using $p_c^2=r^2-p^2$ we get:
$$
\tilde p=\frac{R p}{r^2},
$$
which implies
$$
f(r,p)=0 \qquad \stackrel{I_R}{\longrightarrow}\qquad f\zav{\frac{R}{r},\frac{R p}{r^2}}=0,
$$
as claimed.

A generalization of this argument gives the next theorem:
\begin{theorem}\label{transth}
Let $T$ be a transform that maps a curve in polar coordinates $(r,\varphi)=0$ into the curve $(\tilde r,\tilde \varphi)$ as follows
\begin{align*}
\tilde r&= f(r),\\
\tilde \varphi &= \inte{\varphi_0}{\varphi} g( r){\rm d}t.
\end{align*}
Then the transform $T$ in pedal coordinates takes form
\begin{align*}
h\zav{r,\frac{1}{p^2}}&=0 & &\stackrel{T}{\longrightarrow} & h\zav{f(r),\zav{\frac{f^\prime(r) r^2}{f(r)^2 g(r)}}^2\zav{\frac{1}{p^2}-\frac{1}{r^2}}+\frac{1}{f(r)^2}}=0,
\end{align*}
\end{theorem}
\begin{proof}
$$
f^2\zav{\frac{f^2}{\tilde p^2}-1}=\zav{\frac{f \tilde p_c}{\tilde p}}^2=\zav{\frac{\tilde r \tilde p_c}{\tilde p}}^2=\zav{\frac{\partial \tilde r}{\partial\tilde\varphi}}^2=\zav{\frac{\partial f}{g\partial \varphi}}^2=\frac{{f^\prime}^2 \zav{r^\prime_\varphi}^2}{g^2}=\zav{\frac{f^\prime}{g}}^2 \zav{\frac{rp_c}{p}}^2=\zav{\frac{f^\prime r^2}{g}}\zav{\frac{r^2}{p^2}-1}.
$$
Solving for $\frac{1}{\tilde p^2}$ we obtain the result.
\end{proof}

By means of this theorem many transforms can be easily obtained. All the transforms listed in (\ref{translist})-(\ref{translistp}) and also (\ref{harmtrans}) are left to the reader as an excercise. For example, the harmonic transform $H_k$ in polar coordinates is given by
$$
\tilde r=r,\qquad \tilde \varphi=\frac{1}{k} \varphi,
$$
and so on.

The only exceptional transform which does not behave so well in polar coordinates is $E_c$ but since parallel curves share the normal, the distance to the normal $p_c$ is conserved and the distance to tangent $p$ is hence just shifted by $c$.

Relations between $F_c$ and $E_c$ and $E^\star_c$:
$$
F_c=P E_c P^{-1},\qquad E^\star_c=D_1 E_c D_1=I_1 F_c I_1,
$$
are then a consequence which can be verified in pedal coordinates by a little bit of algebra. 
The same goes for the properties of $H_k$:
\begin{align*}
H_{\pm 1} &= Id, & H_k H_l &= H_{kl},& H_k I_R &= I_R H_k, & F_c H_k&= H_k F_c,
\end{align*}
and so on.

With this theorem we can also see that the transform (\ref{Mahomedtr}) assumed in paper \cite{Mahomed} which generalized Newton theorem of revolving orbits is the composition of scaling $S_a$, dual parallel $E^\star_b$ and harmonics $H_k$:
$$
(\ref{Mahomedtr})=H_k E^\star_{b}S_a.
$$

Also we can see that the transform $T_f$ in Theorem \ref{nonlocalrevolvingorbits} is a special case of Theorem \ref{transth} for 
$$
g=\frac{f^\prime r^2}{f^2 k}.
$$
Hence we have
\begin{align}\label{Tftransform}
h\zav{r,\frac{1}{p^2}}&=0 & &\stackrel{T_f}{\longrightarrow} & h\zav{f(r),\frac{k^2}{p^2}-\frac{k^2}{r^2}+\frac{1}{f(r)^2}}=0.
\end{align}

%
%

\begin{example}
The circle inverse can be seen as a specific case of a more general \emph{complex power transform} (denoted here $M_{\alpha}$) which acts in polar coordinates as
$$
\tilde r= r^\alpha, \qquad \tilde \varphi=\alpha \varphi=\inte{0}{\varphi} \alpha {\rm d}t,
$$
for a given real number $\alpha$. (The case $\alpha=-1$ corresponds to the circle inverse.)

Using Theorem \ref{transth} we have
$$
h(r,p)=0 \qquad \stackrel{M_\alpha}{\longrightarrow}\qquad h\zav{r^\alpha,r^{\alpha-1}p}=0.
$$
It is easy to check that 
\begin{align*}
M_1&=Id, & M_{-1}&=I_{1}, & M_\alpha M_\beta &= M_{\alpha\beta}. & 
\end{align*} 
\end{example}
\begin{example}
One advantage of pedal coordinates over the polar ones is that even a nonlocal change in the $\varphi$ variable translates algebraically into pedal coordinates.

Take the transform $A_\omega$, for example, introduced in Section \ref{Intro}:
$$
\tilde r= r,\qquad \tilde \varphi= \varphi- \varphi_0+\omega\inte{ \varphi_0}{ \varphi} r^2(t){\rm d}t.
$$
By Theorem \ref{transth} we have 
\begin{equation}\label{Aomega}
f\zav{r,\frac{1}{p^2}}=0,\qquad \stackrel{A_\omega}{\longrightarrow} \qquad f\zav{r,\zav{\frac{1}{p^2}+2\omega +\omega^2 r^2}\frac{1}{(1+\omega r^2)^2}}=0.
\end{equation}

Combining this transform with complex powers $B_\alpha:= M_\frac12 A_\alpha M_2$,
\begin{equation}\label{Balpha}
f\zav{r,\frac{1}{p^2}}=0, \qquad \stackrel{B_\omega}{\longrightarrow} \qquad
f\zav{r,\zav{\frac{1}{p^2}+\frac{2\alpha}{r}+\alpha^2}\frac{1}{(1+\alpha r)^2}},
\end{equation}
we get another interesting nonlocal transform which will be important later. It can also be understood in polar coordinates in the sense of Theorem~\ref{transth}:
$$
\tilde r= r,\qquad \tilde \varphi= \varphi +\alpha \int r{\rm d} \varphi.
$$

\end{example}
\section{Evolute and related transforms}\label{evolutesec}
\subsection{Evolute}
Remember that for a curve $\gamma$ its evolute $E(\gamma)$ is defined to be the locus of centers of osculating circles. It also known that the normal of $\gamma$ becomes tangent line for the evolute (so alternatively, the evolute can be defined to be the locus of normal lines).
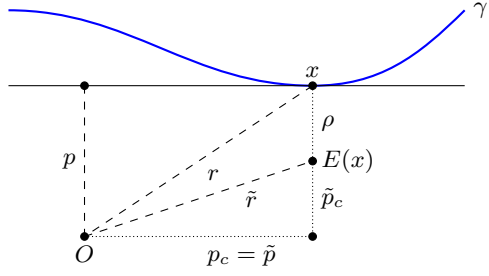
\begin{center}
{\small
\begin{tikzpicture}[domain=-4:5]
\draw [thick,blue] (-4,3) to [out=0,in=180] (0,2) to [out=0, in=225] (2,3);
\coordinate [label=right:{$\gamma$}] (M1) at (2,3);
\coordinate [label=above:{$x$}] (M2) at (0,2);
\coordinate [label=right:{$E(x)$}] (M2) at (0,1);
\draw [thin] (-4,2) to (2,2);
\draw [dashed] (-3,2) to (-3,0);
\draw [dashed] (0,2) to (-3,0);
\draw [dashed] (-3,0) to (0,1);
\filldraw (-3,0) circle (0.05);
\filldraw (0,2) circle (0.05);
\filldraw (-3,2) circle (0.05);
\filldraw (0,1) circle (0.05);
\coordinate [label=below:{$O$}] (M3) at (-3,0);
\draw [densely dotted] (0,2) to (0,0);
\draw [densely dotted] (-3,0) to (0,0);
\filldraw (0,0) circle (0.05);
\coordinate [label=left:{$p$}] (M6) at (-3,1);
\coordinate [label=below right:{$r$}] (M7) at (-1.5,1);
\coordinate [label=right:{$\tilde r$}] (M9) at (-1.0,0.5);
\coordinate [label=below right:{$p_c=\tilde p$}] (M8) at (-1.5,0);
\coordinate [label=right:{$\rho$}] (M10) at (0,1.5);
\coordinate [label=right:{$\tilde p_c$}] (M11) at (0,0.5);
\end{tikzpicture}
}
\end{center}
Since the radius of curvature $\rho$ of a curve with pedal coordinates $(p,r)$ is given by
$$
\rho=rr^\prime,
$$
the pedal coordinates $(\tilde p,\tilde r)$ of the evolute are
\begin{align*}
\tilde r&=\sqrt{p_c^2+\zav{rr^\prime-p}^2}, & &\text{or} & \tilde p_c&= p_c p_c^\prime \\
\tilde p&=p_c, & & & \tilde p&=p_c.
\end{align*}

We can also work out the derivatives
$$
\tilde p_c \tilde p_c^\prime=\tilde p_c\frac{\partial \tilde p_c}{\partial \tilde p}=p_c p_c ^\prime\frac{\partial p_c p_c^\prime}{\partial p_c}=(p_c p_c)^\prime p_c.
$$
Generally
$$
\zav{\tilde p_c\partial_{\tilde p}}^n \tilde p=\zav{ p_cp_c^\prime \partial_{p_c}}^n p_c=\zav{ p_c \partial_{p}}^n p_c=\zav{ p_c \partial_{p}}^{n+1} p.
$$
Hence we arrive at the following proposition:
\begin{proposition}\label{evolute} The evolute $E(\gamma)$ of a curve $\gamma$ which satisfies
$$
f\zav{p_c,p_c p_c^\prime,\zav{p_c p_c^\prime}^\prime p_c,\dots, \zav{p_c\partial_p}^n p}=0,
$$
where $n>1$, satisfies
$$
f\zav{p,p_c,p_c p_c^\prime,\zav{p_c p_c^\prime}^\prime p_c,\dots, \zav{p_c\partial_p}^{n-1} p}=0.
$$
In other words
$$
f\zav{p_c,p_c p_c^\prime,\dots, \zav{p_c\partial_p}^n p}=0,\qquad \stackrel{E}{\longrightarrow} \qquad f\zav{p,p_c,p_c p_c^\prime,\dots, \zav{p_c\partial_p}^{n-1} p}=0.
$$
\end{proposition}
\begin{example}
Since the circle with the pedal point at center can be described
$$
p_c=0,
$$
its involute (i.e. inverse of the evolute) is by the previous theorem
$$
p_c=0 \qquad \stackrel{E^{-1}}{\longrightarrow}\qquad p_cp_c^\prime=0 \qquad \Rightarrow\qquad p_c^2=a^2 \qquad \Rightarrow\qquad p_c=a,
$$
for some constant $a\geq 0$ as claimed.

We can also work out the case of the circle with general position of the pedal point:
$$
p_c^2+\zav{p-R}^2=a^2 \qquad \stackrel{E^{-1}}{\longrightarrow}\qquad
\zav{p_cp_c^\prime}^2+\zav{p_c-R}^2=a^2.
$$
This separable ODE has in addition to the solution above also the solution
$$
p+\sqrt{a^2-(p_c-R)^2}-R\arctan\frac{p_c-R}{\sqrt{a^2-(p_c-R)^2}}=c,
$$
showing once again extreme dependence of pedal coordinates on the position of the pedal point.
\end{example}
\subsection{Contrapedal}
It is known that the contrapedal curve $P_c$ defined in the beginning is the pedal of evolute, i.e. $P_c:=PE$ (see \cite[p. 151]{Zwikker}). Thus we have:
\begin{corollary}
$$
f\zav{p_c,p_c p_c^\prime,\dots, \zav{p_c\partial_p}^n p}=0,\qquad \stackrel{P_c:=P E}{\longrightarrow} \qquad P\zav{f\zav{p,p_c,p_c p_c^\prime,\dots, \zav{p_c\partial_p}^{n-1} p}=0},
$$
or using Proposition \ref{MPC}:
$$
f\zav{p_c,p_c p_c^\prime,\dots, \zav{p_c\partial_p}^n p}=0\qquad \stackrel{P_c}{\longrightarrow} \qquad f\zav{r,\abs{r^\prime_\varphi},r^{\prime\prime}_{\varphi},\dots, r^{(n-1)}_\varphi}=0.
$$
Equivalently, we can say:
$$
f\zav{\abs{r^\prime_\varphi},r^{\prime\prime}_{\varphi},\dots, r^{(n)}_\varphi}=0 \qquad \stackrel{P E P^{-1}}{\longrightarrow} \qquad f\zav{r,\abs{r^\prime_\varphi},r^{\prime\prime}_{\varphi},\dots, r^{(n-1)}_\varphi}=0.
$$
\end{corollary}
\subsection{Catacaustic}
For a given curve $\gamma$ the ``catacaustic curve'' $C$ is defined to be the locus of lines reflected by $\gamma$ originating from a given point (the so-called ``radiant'' which, for our purposes, will coincide with the pedal point). It is known (\cite[p. 60 and 207]{Lawrence}) that catacaustic is the same as the evolute of the orthonomic curve (which is the pedal curve magnified by factor 2), hence:
\begin{corollary} The catacaustic curve $C$ with radiant at the pedal point is given by 
$$
C=E S_\frac12 P.
$$
\end{corollary}
\begin{example}
One of the most beautiful features of a conic section is that lines coming from one focus are reflected to the other focus.

It is hence obvious that the catacaustic curve of a conic section with the radiant(= pedal) at one focus is a single point (the other focus). Reversing this logic, the \emph{anticatacaustic} $C^{-1}=P^{-1} S^{-1}_{\frac12} E^{-1}=P^{-1}S_{2}E^{-1}$ of a point contains conic sections (and, as we will see, nothing more).

Since the pedal equation of a point is $r=a$ or $p_c^2+p^2=a^2$ (i.e. a circle with zero radius), we thus have:
\begin{align*}
C^{-1}\zav{p_c^2+p^2=a^2}&\Rightarrow P^{-1} S_2 E^{-1}\zav{p_c^2+p^2=a^2}\\
&\Rightarrow P^{-1}S_2 \zav{(p_cp_c^\prime)^2+p_c^2=a^2} \\
&\Rightarrow P^{-1}S_2 \zav{(p-2R)^2+p_c^2=a^2} \\
&\Rightarrow P^{-1} \zav{(p-R)^2+p_c^2=\frac{a^2}{4}} \\
&\Rightarrow P^{-1} \zav{ 2R p=r^2+R^2-\frac{a^2}{4}} \\
&\Rightarrow 2R\frac{p^2}{r}=p^2+R^2-\frac{a^2}{4}\\
&\Rightarrow \frac{R^2-\frac{a^2}{4}}{p^2}= \frac{2R}{r}-1.
\end{align*}
Notice that the third equation from the end informs us that a conic section is the antipedal of a circle, giving us yet another method of construction. But since it holds
$$
P^{-1}= I_RPI_R,
$$
which can be easily verified, this construction is actually very close to the one from Example \ref{Kp}.
\end{example}
%

\section{$f$ spirals}\label{spiralssec}

Transforms derived in previous two sections allow us to construct large quantity of curves. It is advantageous to collect them into families.
\subsubsection{Sinusoidal spiral} The family of curves $\sigma_n(a)$, given in polar and pedal coordinates
$$
r^n=a^n\sin\zav{n\varphi+\varphi_0}, \qquad a^n p =r^{n+1},
$$
respectively, contains many famous curves, e.g.
\begin{align*}
n& & a^n p &=r^{n+1} & &\text{Curve} &\text{Pedal point:}\\
n&=0 & p&=r & &\text{Concentric circle }\abs{x}=R. & \text{Center.}\\
n&=-1 & p&=a & &\text{Line.} & \text{A point distant }a.\\
n&=1 & a p&=r^2 & &\text{Circle. } & \text{On the circle.}\\
n&=2 & a^2 p &=r^3 & &\text{Lemniscate of Bernoulli.} & \text{Center.}\\
n&=-2 & rp &=a^2 & &\text{Rectangular hyperbola.} & \text{Center.}\\
n&=-\frac12 & a^{-\frac12} p &=r^{\frac12} & &\text{Parabola.} & \text{Focus}\\
n&=\frac12 & a^{\frac12} p &=r^{\frac32} & &\text{Cardioid.} & \text{Cusp.}
\end{align*}
This family, called ``sinusoidal spirals'', is famously invariant under a number of transforms, for example: 
\begin{align*}
\sigma_n(a) &\stackrel{S_\alpha}{\longrightarrow} \sigma_{n}\zav{\frac{a}{\alpha}} & \text{Scaling.}\\
\sigma_n(a) &\stackrel{P}{\longrightarrow} \sigma_{\frac{n}{n+1}}(a) & \text{Pedal.}\\
\sigma_n(a) &\stackrel{I_R}{\longrightarrow} \sigma_{-n}\zav{\frac{R}{a}} & \text{Inverse.}\\
\sigma_n(a) &\stackrel{M_\alpha}{\longrightarrow} \sigma_{\alpha n}\zav{a^{\frac{1}{\alpha}}} & \text{Complex power.}
\end{align*}
Combining scaling with complex power, we get 
$$
\sigma_1(a) \qquad \stackrel{S_{\beta}M_\alpha}{\longrightarrow} \qquad \sigma_{\alpha}(a),\qquad \beta:=a^\frac{1-\alpha}{\alpha},
$$
which shows that \emph{all} sinusoidal spirals are (up to scaling) complex powers of a circle passing through the origin $\sigma_1(a)$.

\subsubsection{Spirals}
A similar argument can be made about another family of curves $\varsigma_\alpha(c)$:
$$
r=c\varphi^\alpha \qquad \frac{1}{p^2}=\frac{\alpha^2 c^{\frac{2}{\alpha}}}{r^{2+\frac{2}{\alpha}}}+\frac{1}{r^2}.
$$
Specific cases include among others:
\begin{align*}
\alpha&=1 & \frac{1}{p^2}&=\frac{1}{r^2}+\frac{c^2}{r^4} &\text{Spiral of Archimedes}\\
\alpha&=-1 & \frac{1}{p^2}&=\frac{1}{r^2}+\frac{1}{c^2} &\text{Hyperbolic spiral}\\
\alpha&=\frac12 & \frac{1}{p^2}&=\frac{1}{r^2}+\frac{c^4}{4 r^6} &\text{Fermat spiral}\\
\alpha&=-\frac12 & \frac{1}{p^2}&=\frac{1}{r^2}+\frac{r^2}{4 c^4} &\text{Lituus.}
\end{align*}
Similarly, they are also invariant under complex powers and scaling
\begin{align*}
\varsigma_\alpha(c) &\stackrel{S_\gamma }{\longrightarrow} \varsigma_{\alpha}\zav{\frac{c}{\gamma}} & \text{Scaling.}\\
\varsigma_\alpha(c)&\stackrel{H_k}{\longrightarrow}\varsigma_{\alpha}\zav{\frac{c}{k^\alpha}} &\text{Harmonic is scaling.}\\ 
\varsigma_\alpha(c) &\stackrel{ M_\beta}{\longrightarrow} \varsigma_{\alpha\beta}\zav{\beta^{\frac{\alpha}{\beta}}c^{-\frac{1}{\beta}}} & \text{Complex power.}\\
\end{align*}
Which demonstrates that all of these spirals are complex powers of the spiral of Archimedes (for example; starting curve can be of course any spiral $\varsigma_\alpha(c)$).

\subsubsection{General spirals}
There is a pattern in previous two examples which can be generalized yielding more interesting families of curves. 
We start with a curve given in polar coordinates as 
$$
r=f\zav{\varphi+\varphi_0},
$$
for some well behaved function $f$.
Since the pedal coordinates are oblivious to any rotation, the phase factor $\varphi_0$ can be chosen arbitrarily.

Now we make all complex powers of this curve:
$$
r^\alpha=f\zav{ \alpha\varphi+\varphi_0},
$$
and harmonics
$$
r^\alpha=f\zav{l\alpha\varphi+\varphi_0},
$$
and scalings and we end up with:
\begin{definition} Given function $f$ the \emph{ $f$-spirals} is the family of curves given in polar coordinates by
$$
r^\alpha=\frac{c}{l} f\zav{ l\alpha\varphi+\varphi_0},
$$
where $\alpha,l,c,\varphi_0\in\mathbb{R}$, $l\neq 0$ are any real numbers so that the right hand side of this equation defines a real function.
\end{definition}
\begin{remark}
To get the pedal equation for this family it suffice to translate the original curve 
$$
r=f(\varphi),
$$ 
into pedal coordinates and then perform the transform: $S_{\zav{\frac{l}{c}}^{\frac{1}{\alpha}}} H_{\frac{1}{l}}M_\alpha$.
\end{remark}
Few examples are the following:
\begin{align}
f&=\exp & p&=\frac{l}{\sqrt{1+l^2}} r &\text{logarithmic spiral}\\
f&=Id & \frac{1}{p^2}&=\frac{1}{r^2}+\frac{l^2c^{2\alpha}}{r^{2\alpha+2}} &\text{ spirals}\\
\label{sineq}f&=\sin,\cos & \frac{1}{p^2}&=\frac{1-l^2}{r^2}+\frac{c^2}{r^{2\alpha+2}} &\text{Harmonics of sinusoidal spirals}\\
\label{sinheq}f&=\sinh & \frac{1}{p^2}&=\frac{1+l^2}{r^2}+\frac{c^2}{r^{2\alpha+2}} &\sinh\text{ spirals}\\
\label{cosheq}f&=\cosh & \frac{1}{p^2}&=\frac{1+l^2}{r^2}-\frac{c^2}{r^{2\alpha+2}} &\cosh\text{ spirals}\\
\label{sneq}f&={\rm sn} & \frac{1}{p^2}&=\frac{1-l^2(k^2+1)}{r^2}+l^4k^2c^{-2} r^{2\alpha-2}+\frac{c^2}{r^{2\alpha+2}} &\text{elliptic sinusoidal spirals}\\
\label{cneq}f&={\rm cn} & \frac{1}{p^2}&=\frac{(k^2-{k^\prime}^2)l^2+1}{r^2}-l^4k^2c^{-2} r^{2\alpha-2}+\frac{{k^\prime}^2c^2}{r^{2\alpha+2}} &\text{elliptic cosinusoidal spirals.}
\end{align}

All the $f$ spirals families are invariant (by construction) under scaling, complex power and harmonic transforms but some of them happen to be invariant under larger group of transforms, e.g. logarithmic spirals are famously invariant under pretty much everything -- pedal, contrapedal, evolute, orthoptic, catacaustic, etc. Only parallel curves and involutes of logarithmic spirals are slightly different:
$$
p=\abs{\sin\alpha} r \qquad \stackrel{E_c=E^{-1}}{\longrightarrow} \qquad \abs{\cos\alpha}(p-c)=\abs{\sin\alpha} p_c.
$$
\subsubsection{Sinusoidal spirals}
Deriving the case $f=\sin$ is simple using the fact that the function $\sin\varphi$ is a solution of the differential equation
$$
{r^\prime_\varphi}^2+r^2=1,
$$
which translates into pedal coordinates by Proposition \ref{MPC} as
\begin{equation}
\frac{1}{p^2}=\frac{1}{r^4}\qquad \stackrel{S_{\zav{\frac{l}{c}}^{\frac{1}{\alpha}}} H_{\frac{1}{l}}M_\alpha}{\longrightarrow} \qquad
\frac{1}{p^2}=\frac{1-l^2}{r^2}+\frac{c^2}{r^{2\alpha+2}}.
\end{equation}
Hence the $\sin$-spirals are just harmonics of the usual sinusoidal spirals.

Notice that $\cos\varphi$ satisfies the same differential equation as $\sin\varphi$, hence the choice $f=\cos$ gives us exactly the same result. Thus the $cos$-spirals and the $sin$-spirals are (in pedal coordinates) indistinguishable.

This is a consequence of the fact that $\cos$ function can be obtain from $\sin$ function by simple shift of argument: $\cos\varphi=\sin\zav{\varphi+\frac12\pi}$.

Similarly we can ask what are $sinh$-spirals (i.e. $f=\sinh$). The differential equation is now
$$
{r^\prime_\varphi}^2-r^2=1,
$$
hence obtaining (\ref{sinheq}).
Making a similar argument for the $cosh$-spirals we get (\ref{cosheq}).

This time we cannot get $\cosh$ form $\sinh$ by any (real) shift of the argument so the pedal equations are different.

But using \emph{complex} translations we have: 
$\sinh\zav{\varphi+\frac{\imag\pi}{2}}=\imag\cosh(\varphi)$ which gives us $\cosh$ for purely imaginary scaling factor $c$. (Substituting $c\to \imag c$ in (\ref{sinheq}) gives us (\ref{cosheq}).)

Exploiting the relations:
$$
\imag\sin(\imag \varphi)=\sinh(\varphi),\qquad \sin\zav{\imag\varphi+\frac{\pi}{2}}=\cosh(\varphi),
$$ 
we can see that both $sinh$-spirals (\ref{sinheq}) and $cosh$-spirals (\ref{cosheq}) are derivable from $sin$-spirals (\ref{sineq}). (This time substituting $l\to \imag l$ in (\ref{sineq}) gives us (\ref{sinheq}).)

This connection between these three families is reflected in pedal coordinates by the fact that together they solve the general equation in pedal coordinates:
$$
\frac{1}{p^2}=\frac{a}{r^2}+\frac{b}{r^{2\alpha+2}}, \qquad \forall a,b\in\mathbb{R}, a\not=1, \ b\not=0.
$$

Indeed, for $a<1, b>0$ we have $sin$-spirals.
For $a>1, b>0$ we have $sinh$-spirals and for $a>1, b<0$ the $cosh$-spirals. The case $a<1, b<0$ does not define any curve in pedal coordinates since in this case it holds $p>r$, which is impossible.

The limiting cases $a=1$ are just $Id$-spirals and $b=0$ gives us logarithmic spirals. 
\subsubsection{Elliptic spirals}
We are going now to make similar argument for elliptic version of sinusoidal spirals (i.e. ``snusoidal spiral'').

Remember that there are 12 Jacobian elliptic functions:
$$
{\rm sn}(z),{\rm cn}(z), {\rm dn}(z), {\rm ns}(z), {\rm nc}(z), {\rm nd}(z), {\rm sc}(z), {\rm sd}(z), {\rm cs}(z), {\rm ds}(z), {\rm cd}(z), {\rm dc}(z),
$$
all of which depend on an additional parameter $k\in (-1,1)$ (the so-called ``modulus'') which we will not explicitly mention. Those function are doubly periodic with periods $4K, 4\imag K^\prime $, where $K\equiv K(k)$ (so-called \emph{quarter period}) is the complete elliptic integral of the first kind
$$
K(k):=\inte{0}{\frac{\pi}{2}}\frac{1}{\sqrt{1-k^2\sin\varphi}}{\rm d}\varphi=\frac{\pi}{2}\!\! \ _2 F_1\zav{\nadsebou{\frac12\quad \frac12}{1};k^2},
$$
and $k^\prime:=\sqrt{1-k^2}$, $K^\prime:=K(k^\prime)$.

These functions are connected by the formula
$$
pq(z)=\frac{pr(z)}{qr(z)},
$$
where $p,q,r$ can be any of the letters $s,c,d,n$. In particular $pp(z)=1$ and $pq(z)=\frac{1}{qp(z)}$.

The parameter $k$ can be analytically continued beyond the interval $(-1,1)$. In fact, for fixed $z$, all elliptical functions are meromorphic with respect to $k^2$.

With this in mind we can construct a $13^{th}$ elliptical function (which we denote sn$^\star$) which is just sn function with modulus $k$ restricted to the unit circle in the complex plane, more precisely:
$$
{\rm sn}^\star(z;\lambda):=e^{\imag \frac{\lambda}{2}} {\rm sn}\zav{e^{-\imag \frac{\lambda}{2}}z, e^{\imag\lambda}}.
$$
The property of the sn function:
$$
{\rm sn}(z,k)= \frac{1}{k}{\rm sn}\zav{zk,\frac{1}{k}},
$$
ensures that ${\rm sn}^\star$ is a real valued function (for real argument). No other function than sn has analogous property (save $ns$) and without this function the picture will be incomplete, as we will see.

Let us start with sn function. For the pedal equation of $sn$-spirals
$$
r^\alpha=\frac{c}{l}{\rm sn}\zav{l\alpha\varphi+\varphi_0},
$$
we make use of the differential equation valid for $r={\rm sn}(\varphi)$
$$
{r^\prime_\varphi}^2=(1-r^2)(1-k^2 r^2),
$$ 
which translates into pedal coordinates by Proposition \ref{MPC} as
$$
\frac{r^4}{p^2}-r^2=(1-r^2)(1-k^2r^2),
$$
and making the transform $S_{\zav{\frac{l}{c}}^{\frac{1}{\alpha}}} H_{\frac{1}{l}}M_\alpha$ we get (\ref{sneq}).

Similar argument can be made for all the remaining 12 elliptic function but, fortunately, we can exploit many known relations to simplify the job: 
\begin{align}\label{ellrel}
{\rm ns}(z)&=k\ {\rm sn}(z+\imag K^\prime), &
{\rm sn}^\star(z)&:=e^{\imag \frac{\lambda}{2}} {\rm sn}\zav{e^{-\imag \frac{\lambda}{2}}z, e^{\imag\lambda}}
\\ \nonumber
{\rm cd}(z)&={\rm sn}(z+K), &
{\rm dc}(z)&=k\ {\rm sn}(z+K+\imag K^\prime)\\ \nonumber
{\rm dn}(z)&=k^\prime\ {\rm sn}(\imag z+K^\prime+\imag K,k^\prime), &
{\rm nd}(z)&={\rm sn}(\imag z+K^\prime,k^\prime),\\ \nonumber
{\rm cs}(z)&=\imag k^\prime\ {\rm sn}(\imag z+\imag K,k^\prime), &
{\rm sc}(z)&=-\imag\ {\rm sn}(\imag z,k^\prime),\\ \nonumber
{\rm sd}(z)&=\frac{1}{k^\prime} {\rm sn}\zav{zk^\prime,\frac{\imag k}{k^\prime}}, &
{\rm ds}(z)&=\imag k\  {\rm sn}\!\zav{zk^\prime+Kk^\prime+\imag K^\prime k^\prime,\frac{\imag k}{k^\prime}},\\ \nonumber
{\rm cn}(z)&= {\rm sn}\zav{zk^\prime+Kk^\prime,\frac{\imag k}{k^\prime}}, &
{\rm nc}(z)&=-\imag\ {\rm sn}\!\zav{zk^\prime+\imag K^\prime k^\prime,\frac{\imag k}{k^\prime}}.
\end{align} 
This table tells us two things: first, the functions ${\rm sn}\equiv {\rm cd}\equiv{\rm dc}\equiv{\rm ns}$ are equivalent in the sense that they generate the same family of $f$-spirals in pedal coordinates (since they differ only by scaling and a shift of the argument). In the same way the functions ${\rm cn}\equiv{\rm sd}$, ${\rm dn}\equiv{\rm nd}$, ${\rm nc}\equiv{\rm ds}$ and ${\rm sc}\equiv{\rm cs}$ are equivalent. The function sn$^\star$ stands alone. Hence we have only 6 distinct $f$-spirals (out of 13 Jacobian elliptic functions).

Second, the pedal equation of $f$-spirals for all elliptic functions can be obtained from the $sn$-spiral case. For example, the relation for cn function informs us that the $cn$-spirals (\ref{cneq}) can be obtained from the $sn$-spirals (\ref{sneq}) substituting $k\to \imag \frac{k}{k^\prime} $, $l\to k^\prime l$, $c\to k^\prime c$.

Similarly, ${\rm sn}^\star$-spirals:
$$
\frac{1}{p^2}=\frac{1-2l^2\cos\lambda}{r^2}+\frac{l^4}{c^2}r^{2\alpha-2}+\frac{c^2}{r^{2\alpha+2}}.
$$
are obtained by substituting $k\to e^{\imag\lambda} $, $l\to l e^{-\imag\frac{\lambda}{2}}$, in (\ref{sneq}).
And so on.

Once again this interconnectedness is reflected in pedal coordinates by the fact that together they solve the general equation:
\begin{equation}\label{ellspirals}
\frac{1}{p^2}=\frac{a}{r^2}+\beta r^{2\alpha-2}+\frac{\gamma}{r^{2\alpha+2}}. 
\end{equation}

Specifically:
\begin{align*}
a&<1-4\beta\gamma & \beta&>0 & \gamma&>0 & \text{ $sn$-spirals,}\\
a&>1+4\beta\gamma & \beta&>0 & \gamma&>0 & \text{ $sc$-spirals,}\\
(a-1)^2&\leq 4\beta \gamma & \beta&>0 & \gamma&>0 & \text{ $sn^\star$-spirals,} \\
& & \beta&<0 & \gamma&>0 & \text{ $cn$-spirals,}\\
& & \beta&>0 & \gamma&<0 & \text{ $ds$-spirals,}\\
(a-1)^2&\geq 4\beta\gamma & \beta&<0 & \gamma&<0 & \text{ $dn$-spirals,}\\
(a-1)^2&<4\beta\gamma & \beta&<0 & \gamma&<0 & \text{ no curve ($p>r$).}
\end{align*}

For $sn$-spiral, the exact solution is 
$$
c=\sqrt{\gamma},\qquad l^2=\sqrt{(1-a)^2-4\beta\gamma}+1-a,\qquad k=\frac{2\sqrt{\beta\gamma}}{\sqrt{(1-a)^2-4\beta\gamma}+1-a}.
$$
It is straightforward from (\ref{ellrel}) to make similar computation for all elliptic spirals.

The limiting case $\beta\gamma=0$ is just the $sin$-spiral (and its complex extensions) treated above.


\section{Central and Lorentz like force problems}\label{CentralForce}
We are ready to prove Theorem \ref{cfth}.
%

\begin{proof}
Making the scalar product of the equation (\ref{dynsys}) with $\dot x$ we obtain 
$$
\ddot x \cdot \dot x=F^\prime\zav{\abs{x}^2}x\cdot \dot x, 
$$
integrating we get the first conserved quantity
$$
\abs{\dot x}^2=F\zav{\abs{x}^2}+c.
$$
Similarly, making the scalar product with $x^\perp$ we get
$$
\frac{{\rm d}}{{\rm d} t}\zav{{\dot{x}} \cdot x^\perp}=\ddot x \cdot x^\perp=2 G^\prime\zav{\abs{x}^2}{\dot x}^\perp\cdot x^\perp= 2 G^\prime\zav{\abs{x}^2}{\dot x}\cdot x,
$$
with the integral
$$
\dot x\cdot x^\perp=-x\cdot {\dot x}^{\perp}=G\zav{\abs{x}^2}-L.
$$
The pedal coordinates $p$ (distance to the tangent vector) and $r$ (distance from the origin) are given by
$$
p=\frac{x\cdot {\dot x}^\perp}{\abs{\dot x}},\qquad r=\abs{x}.
$$
Hence we have
$$
p^2=\frac{\zav{G\zav{r^2}-L}^2}{F\zav{r^2}+c},
$$
as claimed. The inequality (\ref{regine}) is the consequence of the fact that
$$
p\leq r.
$$

\end{proof}
\begin{remark}
For the strictly central force case (i.e. $G\equiv 0$) we might equivalently say that ``the potential energy is inversely proportional to the square of the distance to the tangent'', i.e.
$$
F \propto\frac{1}{p^2}.
$$
\end{remark}
\begin{example}
%
%
At the time when it was still an open issue, it was suggested that orbit of planets are Cassini ovals, with Sun at one focus. The person who suggested it was, allegedly, Cassini himself and we are now in the position to see what force law we must assume for him to be right.

Remember that a Cassini oval is locus of points such that product of distances from two foci is constant, i.e.
$$
\abs{x}\abs{x-a}=C.
$$ 
Pedal form of this equation is easy to show to be (by an analogous argument as in Example \ref{Covalex}):
$$
\frac{\zav{3C^2+r^4-\abs{a}^2 r^2}^2 }{p^2}=4C^2\zav{\frac{2C^2}{r^2}+2r^2-\abs{a}^2}. 
$$
Hence the corresponding dynamical system by Theorem \ref{cfth} looks like
$$
\ddot x=\zav{8C^2 -\frac{16C^4}{\abs{x}^4}}x+\zav{\abs{x}^2-\abs{a}^2}\dot x^\perp,
$$
quite different from the gravitational inverse square law.

We can even work the case of a Cassini oval with Sun at the center, i.e.
$$
\abs{x-a}\abs{x+a}=C, \qquad \text{or}\qquad \abs{x^2-a^2}=C,
$$
where quantities $x,a$ are treated as complex numbers. This is, obviously, a (complex) square of a circle, i.e.
$$
\abs{x-a^2}=C \qquad \stackrel{M_{2}}{\longrightarrow} \qquad \abs{x^2-a^2}=C.
$$
Hence its pedal equation is
$$
2Rp=r^2+R^2-a^2, \qquad \stackrel{M_2}{\longrightarrow}\qquad 2R pr=r^{4}+R^2-a^2,
$$
or
$$
\frac{\zav{r^{4}+R^2-a^2}^2}{p^2}=4R^2 r^2.
$$
\end{example}

More challenging are inverse questions which we will tackle in some cases.
\begin{example}
Revisiting the Kepler problem from Example \ref{Kp}, 
$$
\ddot{x}=-\frac{M}{\abs{x}^{3}}x,
$$
we can arrive at the solution immediately without the need of polar coordinates:
$$
\frac{L^2}{p^2}=\frac{2M}{r}+c,
$$
albeit, apparently, without the information about the construction. But with a quite easy observation, making the transform (\ref{translistp}) $E^\star_{\alpha}$ with $\alpha:=-\frac{M}{L^2}$ we obain
$$
\frac{L^2}{p^2}=\frac{M^2}{L^2}+c,
$$
which is a line distant $\frac{L^2}{\sqrt{M^2+\frac{c}{L^2}}}$. Hence we have discovered that a conic section is dually parallel to a line -- which is, under close inspection, the same construction as in Example \ref{Kp}.

The dual curve $D_1$ of the solution is easy to see to be a circle:
$$
L^2r^2=2Mp+c.
$$ 
Hence we have recovered the famous Newton result on the solution's curve of velocities. This observation is usually derived by studying the Runge-Laplace-Lenz vector -- a conserved quantity we even do not need. In pedal coordinates this amounts to taking the dual curve.
\end{example}

Proof of Theorem \ref{nonlocalrevolvingorbits}:
\begin{proof}
For the sake of simplicity we write the force $F$ as $F(r)=r F^\prime(r^2)$.
\begin{align*}
\ddot x&=F^\prime(r^2)x \\
&\Downarrow Th \ref{cfth}\ (L=\tilde L)\\
\frac{\tilde L^2}{p^2}&= F(r^2)+ c \\
&\downarrow \ T_f \ (\ref{Tftransform}) \\
\frac{(\tilde Lk)^2}{p^2}&= F\zav{f^2}+\frac{\tilde L^2 k^2}{r^2}-\frac{\tilde L^2}{f^2}+ c \\
&\Uparrow Th \ref{cfth}\ (L=\tilde Lk)\\
\ddot x&=\zav{ f f^\prime F^\prime (f^2)-\frac{\tilde L^2 k^2}{r^3}+\frac{\tilde L^2 f^\prime}{f^3}}\frac{x}{r}.
\end{align*}
\end{proof}

\subsection{Kepler problem in General relativity}\label{RKP}
We can make the same analysis for many problems. Particularly interesting is the problem of orbits (or geodesics) around a non-rotating compact body described by the Schwarzschild solution of Einstein equations of General relativity \cite{schwarzschild}:
\begin{equation}\label{KPGR}
{r^\prime_\varphi}^2=\frac{r^4}{b^2}-\zav{1-\frac{r_s}{r}}\zav{\frac{r^4}{a^2}+a^2},
\end{equation}
where
$$
r_s:=\frac{2G M}{c^2}, \qquad a:=\frac{L}{GM c}, \qquad b:=\frac{cL}{E}.
$$
The quantity $r_s$ is the Schwarzschild radius and $a,b$ are length-scales introduced for brevity which depends on the angular momentum $L$ and energy $E$ of a test particle.

In pedal coordinates this becomes (using Proposition \ref{MPC}):
\begin{equation}\label{KPGRp}
\frac{1}{p^2}=d+\frac{r_s}{a^2 r}+\frac{r_s}{r^3}, \qquad d:=\frac{1}{b^2}-\frac{1}{a^2}. 
\end{equation}

The second part of Theorem \ref{cfth} informs us that the image of the trajectory is located in the region
$$
\frac{1}{r^2}\leq d+\frac{r_s}{a^2 r}+\frac{r_s}{r^3},
$$
or
$$
0\leq d r^3+\frac{r_s}{a^2} r^2-r+r_s=:h(r).
$$

Since the absolute term $r_s>0$ is positive, the origin is always included in the image (put $r=0$) and this ensures the existence of an unstable component (a component which includes the origin) where the trajectories will reach the origin (a.k.a. singularity). Furthermore this unstable region is at least $r_s$ long since $h(r_s)=\frac{r_s^3}{b^2}\geq 0$.

The overall behavior depends additionally on $N$ -- the number of positive real roots of the polynomial $h$. By the classical result of Fourier and Budan \cite{Fourier,Budan} this numbers is equal to $N= \nu-2\lambda$, where $\lambda\in\mathbb{Z}_+$ and $\nu$ \emph{sign variation} of the polynomial $h$ (i.e. how many times its non-zero coefficients change sign as listed in order of increasing degree).

The list of coefficients of polynomial $h$ is
$$
\zav{r_s,-1,\frac{r_s}{a^2}, d},
$$
hence the sign variation $\nu$ depends only on the sign of $d$
$$
\nu=\left\{ \nadsebou{2}{3}\nadsebou{d\geq 0}{d<0}\right.,
$$
thus the number of zeros is either $N=2,0$ for $d\geq 0$ or $N=3,1$ for $d<0$.

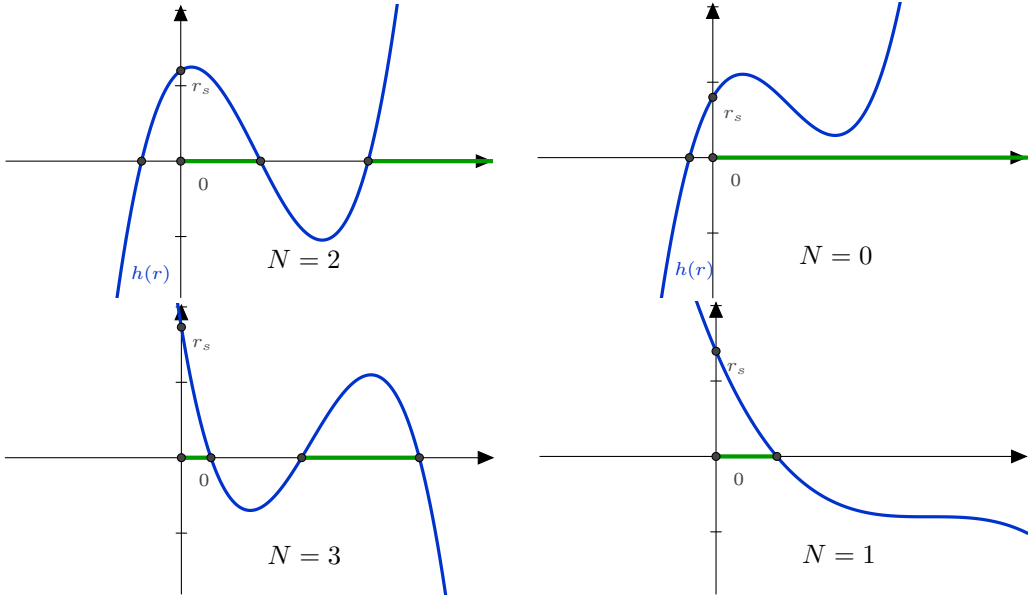
\begin{center}
\definecolor{xdxdff}{rgb}{0.49019607843137253,0.49019607843137253,1.}
\definecolor{qqzzqq}{rgb}{0.,0.6,0.}
\definecolor{uuuuuu}{rgb}{0.26666666666666666,0.26666666666666666,0.26666666666666666}
\definecolor{qqttcc}{rgb}{0.,0.2,0.8}
\begin{tikzpicture}[line cap=round,line join=round,>=triangle 45,x=1.0cm,y=1.0cm]
\draw[->,color=black] (-2.2994212840329946,0.) -- (4.0994728384867365,0.);
\foreach \x in {-2.,-1.,1.,2.,3.,4.}
\draw[shift={(\x,0)},color=black] (0pt,-2pt);
\draw[->,color=black] (0.,-1.808244285033815) -- (0.,2.074989615176037);
\foreach \y in {-1.,1.,2.}
\draw[shift={(0,\y)},color=black] (2pt,0pt) -- (-2pt,0pt);
\clip(-2.2994212840329946,-1.808244285033815) rectangle (4.0994728384867365,2.074989615176037);
\draw[line width=4.pt] (-6.115120681630512,5.468598371446387) -- (-2.73839555101324,5.468598371446387);
\draw[line width=4.pt] (-6.08135343032434,4.827020596629106) -- (-2.7046282997070676,4.827020596629106);
\draw[line width=4.pt] (-6.047586179018166,4.16855919615874) -- (-2.6708610484008943,4.16855919615874);
\draw[line width=4.pt] (-6.013818927711994,3.52698142134146) -- (-2.637093797094722,3.52698142134146);
\draw[line width=1.2pt,color=qqttcc,smooth,samples=100,domain=-2.2994212840329946:4.0994728384867365] plot(\x,{0.9*((\x)-1.0)^(3.0)-((\x)-1.0)-((\x)-1.0)+0.1});
\draw [line width=1.6pt,color=qqzzqq] (0.,0.)-- (1.0500564407024764,0.);
\draw (1.0097693439719326,-1.0484811306449306) node[anchor=north west] {$N=2$};
\draw [line width=1.6pt,color=qqzzqq,domain=2.465053316605841:4.0994728384867365] plot(\x,{(-0.-0.*\x)/3.3734329641487864});
\draw (-2.4513739149107736,-1.808244285033815) -- (-2.4513739149107736,2.074989615176037);
\draw (4.183890966752164,-1.808244285033815) -- (4.183890966752164,2.074989615176037);
\draw [domain=-2.2994212840329946:4.0994728384867365] plot(\x,{(-2.2100586204007278-0.*\x)/-1.});
\draw [domain=-2.2994212840329946:4.0994728384867365] plot(\x,{(--1.9095460389523327-0.*\x)/-1.});
\begin{scriptsize}
\draw [fill=black] (-4.122852854566322,5.468598371446387) circle (2.5pt);
\draw[color=black] (-3.954016598035458,5.865363574293915) node {$a = 0.9$};
\draw [fill=black] (-4.7306633780774305,4.827020596629106) circle (2.5pt);
\draw[color=black] (-4.612477998505826,5.223785799476635) node {$b = -1$};
\draw [fill=black] (-4.325456362403358,4.16855919615874) circle (2.5pt);
\draw[color=black] (-4.156620105872494,4.565324399006268) node {$c = 0.1$};
\draw [fill=black] (-3.9877838493416307,3.52698142134146) circle (2.5pt);
\draw[color=black] (-3.8695984697700263,3.923746624188989) node {$x0 = 1$};
\draw[color=qqttcc] (-0.3915715852342355,-1.495897210451718) node {$h(r)$};
\draw [fill=uuuuuu] (-0.5151097573083177,0.) circle (1.5pt);
\draw [fill=uuuuuu] (1.0500564407024764,0.) circle (1.5pt);
\draw [fill=uuuuuu] (2.465053316605841,0.) circle (1.5pt);
\draw [fill=uuuuuu] (0.,0.) circle (1.5pt);
\draw[color=uuuuuu] (0.30065706654230534,-0.2971597890825898) node {0};
\draw [fill=xdxdff] (5.838486280754627,0.) circle (2.5pt);
\draw [fill=uuuuuu] (0.,1.2) circle (1.5pt);
\draw[color=uuuuuu] (0.2753316280626758,0.9606703220723405) node {$r_s$};
\end{scriptsize}
\end{tikzpicture}
\definecolor{qqzzqq}{rgb}{0.,0.6,0.}
\definecolor{xdxdff}{rgb}{0.49019607843137253,0.49019607843137253,1.}
\definecolor{uuuuuu}{rgb}{0.26666666666666666,0.26666666666666666,0.26666666666666666}
\definecolor{qqttcc}{rgb}{0.,0.2,0.8}
\begin{tikzpicture}[line cap=round,line join=round,>=triangle 45,x=1.0cm,y=1.0cm]
\draw[->,color=black] (-2.2994212840329946,0.) -- (4.13324008979291,0.);
\foreach \x in {-2.,-1.,1.,2.,3.,4.}
\draw[shift={(\x,0)},color=black] (0pt,-2pt);
\draw[->,color=black] (0.,-1.8588951619930738) -- (0.,2.058105989522951);
\foreach \y in {-1.,1.,2.}
\draw[shift={(0,\y)},color=black] (2pt,0pt) -- (-2pt,0pt);
\clip(-2.2994212840329946,-1.8588951619930738) rectangle (4.13324008979291,2.058105989522951);
\draw[line width=4.pt] (-6.115120681630512,5.468598371446387) -- (-2.73839555101324,5.468598371446387);
\draw[line width=4.pt] (-6.08135343032434,4.827020596629106) -- (-2.7046282997070676,4.827020596629106);
\draw[line width=4.pt] (-6.047586179018166,4.16855919615874) -- (-2.6708610484008943,4.16855919615874);
\draw[line width=4.pt] (-6.013818927711994,3.52698142134146) -- (-2.637093797094722,3.52698142134146);
\draw[line width=1.2pt,color=qqttcc,smooth,samples=100,domain=-2.2994212840329946:4.13324008979291] plot(\x,{0.9*((\x)-1.0)^(3.0)+0.0*((\x)-1.0)-((\x)-1.0)+0.7});
\draw (1.0097693439719326,-1.0484811306449306) node[anchor=north west] {$N=0$};
\draw (-2.4513739149107736,-1.8588951619930738) -- (-2.4513739149107736,2.058105989522951);
\draw (4.183890966752164,-1.8588951619930738) -- (4.183890966752164,2.058105989522951);
\draw [domain=-2.2994212840329946:4.13324008979291] plot(\x,{(-2.2100586204007278-0.*\x)/-1.});
\draw [domain=-2.2994212840329946:4.13324008979291] plot(\x,{(--1.9095460389523327-0.*\x)/-1.});
\draw [line width=1.6pt,color=qqzzqq,domain=0.0:4.13324008979291] plot(\x,{(-0.-0.*\x)/5.838486280754627});
\begin{scriptsize}
\draw [fill=black] (-4.122852854566322,5.468598371446387) circle (2.5pt);
\draw[color=black] (-3.954016598035458,5.865363574293915) node {$a = 0.9$};
\draw [fill=black] (-4.392990865015704,4.827020596629106) circle (2.5pt);
\draw[color=black] (-4.325456362403358,5.223785799476635) node {$b = 0$};
\draw [fill=black] (-4.122852854566322,4.16855919615874) circle (2.5pt);
\draw[color=black] (-3.954016598035458,4.565324399006268) node {$c = 0.7$};
\draw [fill=black] (-3.9877838493416307,3.52698142134146) circle (2.5pt);
\draw[color=black] (-3.8695984697700263,3.923746624188989) node {$x0 = 1$};
\draw[color=qqttcc] (-0.2396189543564582,-1.495897210451718) node {$h(r)$};
\draw [fill=uuuuuu] (-0.3063314801980477,0.) circle (1.5pt);
\draw [fill=uuuuuu] (0.,0.) circle (1.5pt);
\draw[color=uuuuuu] (0.30065706654230534,-0.2971597890825898) node {0};
\draw [fill=xdxdff] (5.838486280754627,0.) circle (2.5pt);
\draw [fill=uuuuuu] (0.,0.8) circle (1.5pt);
\draw[color=uuuuuu] (0.2753316280626758,0.555463306398269) node {$r_s$};
\end{scriptsize}
\end{tikzpicture}

\definecolor{qqzzqq}{rgb}{0.,0.6,0.}
\definecolor{xdxdff}{rgb}{0.49019607843137253,0.49019607843137253,1.}
\definecolor{uuuuuu}{rgb}{0.26666666666666666,0.26666666666666666,0.26666666666666666}
\definecolor{qqttcc}{rgb}{0.,0.2,0.8}
\begin{tikzpicture}[line cap=round,line join=round,>=triangle 45,x=1.0cm,y=1.0cm]
\draw[->,color=black] (-2.316304909686081,0.) -- (4.116356464139823,0.);
\foreach \x in {-2.,-1.,1.,2.,3.,4.}
\draw[shift={(\x,0)},color=black] (0pt,-2pt);
\draw[->,color=black] (0.,-1.808244285033815) -- (0.,2.0412223638698648);
\foreach \y in {-1.,1.,2.}
\draw[shift={(0,\y)},color=black] (2pt,0pt) -- (-2pt,0pt);
\clip(-2.316304909686081,-1.808244285033815) rectangle (4.116356464139823,2.0412223638698648);
\draw[line width=4.pt] (-6.115120681630512,5.468598371446387) -- (-2.73839555101324,5.468598371446387);
\draw[line width=4.pt] (-6.08135343032434,4.827020596629106) -- (-2.7046282997070676,4.827020596629106);
\draw[line width=4.pt] (-6.047586179018166,4.16855919615874) -- (-2.6708610484008943,4.16855919615874);
\draw[line width=4.pt] (-6.013818927711994,3.52698142134146) -- (-2.637093797094722,3.52698142134146);
\draw[line width=1.2pt,color=qqttcc,smooth,samples=100,domain=-2.316304909686081:4.116356464139823] plot(\x,{0-0.9*((\x)-1.7)^(3.0)+2.7*((\x)-1.7)-((\x)-1.7)+0.2});
\draw (1.0097693439719326,-1.0484811306449306) node[anchor=north west] {$N=3$};
\draw (-2.4513739149107736,-1.808244285033815) -- (-2.4513739149107736,2.0412223638698648);
\draw (4.183890966752164,-1.808244285033815) -- (4.183890966752164,2.0412223638698648);
\draw [domain=-2.316304909686081:4.116356464139823] plot(\x,{(-2.2100586204007278-0.*\x)/-1.});
\draw [domain=-2.316304909686081:4.116356464139823] plot(\x,{(--1.9095460389523327-0.*\x)/-1.});
\draw [line width=1.6pt,color=qqzzqq] (0.,0.)-- (0.3887344566956177,0.);
\draw [line width=1.6pt,color=qqzzqq] (1.581471357542777,0.)-- (3.1297941857616047,0.);
\begin{scriptsize}
\draw [fill=black] (-4.7306633780774305,5.468598371446387) circle (2.5pt);
\draw[color=black] (-4.511176244587308,5.865363574293915) node {$a = -0.9$};
\draw [fill=black] (-3.4812750797490404,4.827020596629106) circle (2.5pt);
\draw[color=black] (-3.312438823218176,5.223785799476635) node {$b = 2.7$};
\draw [fill=black] (-4.291689111097185,4.16855919615874) circle (2.5pt);
\draw[color=black] (-4.122852854566322,4.565324399006268) node {$c = 0.2$};
\draw [fill=black] (-3.7514130901984215,3.52698142134146) circle (2.5pt);
\draw[color=black] (-3.531925956708299,3.923746624188989) node {$x0 = 1.7$};
\draw[color=qqttcc] (0.0642863073990963,5.916014451253174) node {$h(r)$};
\draw [fill=uuuuuu] (0.3887344566956177,0.) circle (1.5pt);
\draw [fill=uuuuuu] (1.581471357542777,0.) circle (1.5pt);
\draw [fill=uuuuuu] (3.1297941857616047,0.) circle (1.5pt);
\draw [fill=uuuuuu] (0.,0.) circle (1.5pt);
\draw[color=uuuuuu] (0.30065706654230534,-0.2971597890825898) node {0};
\draw [fill=xdxdff] (5.838486280754627,0.) circle (2.5pt);
\draw [fill=uuuuuu] (0.,1.7317) circle (1.5pt);
\draw[color=uuuuuu] (0.2753316280626758,1.5009463429711025) node {$r_s$};
\end{scriptsize}
\end{tikzpicture}
\definecolor{qqzzqq}{rgb}{0.,0.6,0.}
\definecolor{xdxdff}{rgb}{0.49019607843137253,0.49019607843137253,1.}
\definecolor{uuuuuu}{rgb}{0.26666666666666666,0.26666666666666666,0.26666666666666666}
\definecolor{qqttcc}{rgb}{0.,0.2,0.8}
\begin{tikzpicture}[line cap=round,line join=round,>=triangle 45,x=1.0cm,y=1.0cm]
\draw[->,color=black] (-2.316304909686081,0.) -- (4.0994728384867365,0.);
\foreach \x in {-2.,-1.,1.,2.,3.,4.}
\draw[shift={(\x,0)},color=black] (0pt,-2pt);
\draw[->,color=black] (0.,-1.8251279106869012) -- (0.,2.058105989522951);
\foreach \y in {-1.,1.,2.}
\draw[shift={(0,\y)},color=black] (2pt,0pt) -- (-2pt,0pt);
\clip(-2.316304909686081,-1.8251279106869012) rectangle (4.0994728384867365,2.058105989522951);
\draw[line width=4.pt] (-6.115120681630512,5.468598371446387) -- (-2.73839555101324,5.468598371446387);
\draw[line width=4.pt] (-6.08135343032434,4.827020596629106) -- (-2.7046282997070676,4.827020596629106);
\draw[line width=4.pt] (-6.047586179018166,4.16855919615874) -- (-2.6708610484008943,4.16855919615874);
\draw[line width=4.pt] (-6.013818927711994,3.52698142134146) -- (-2.637093797094722,3.52698142134146);
\draw[line width=1.2pt,color=qqttcc,smooth,samples=100,domain=-2.316304909686081:4.0994728384867365] plot(\x,{0-0.1*((\x)-2.8)^(3.0)+((\x)-2.8)-((\x)-2.8)-0.8});
\draw (1.0097693439719326,-1.0484811306449306) node[anchor=north west] {$N=1$};
\draw (-2.4513739149107736,-1.8251279106869012) -- (-2.4513739149107736,2.058105989522951);
\draw (4.183890966752164,-1.8251279106869012) -- (4.183890966752164,2.058105989522951);
\draw [domain=-2.316304909686081:4.0994728384867365] plot(\x,{(-2.2100586204007278-0.*\x)/-1.});
\draw [domain=-2.316304909686081:4.0994728384867365] plot(\x,{(--1.9095460389523327-0.*\x)/-1.});
\draw [line width=1.6pt,color=qqzzqq] (0.,0.)-- (0.8,0.);
\begin{scriptsize}
\draw [fill=black] (-4.4605253676280485,5.468598371446387) circle (2.5pt);
\draw[color=black] (-4.2410382341379265,5.865363574293915) node {$a = -0.1$};
\draw [fill=black] (-4.055318351953977,4.827020596629106) circle (2.5pt);
\draw[color=black] (-3.987783849341631,5.223785799476635) node {$b = 1$};
\draw [fill=black] (-4.629361624158912,4.16855919615874) circle (2.5pt);
\draw[color=black] (-4.4098744906687894,4.565324399006268) node {$c = -0.8$};
\draw [fill=black] (-3.379973325830522,3.52698142134146) circle (2.5pt);
\draw[color=black] (-3.160486192340399,3.923746624188989) node {$x0 = 2.8$};
\draw[color=qqttcc] (-0.7292440982959628,5.916014451253174) node {$h(r)$};
\draw [fill=uuuuuu] (0.8,0.) circle (1.5pt);
\draw [fill=uuuuuu] (0.,0.) circle (1.5pt);
\draw[color=uuuuuu] (0.30065706654230534,-0.2971597890825898) node {0};
\draw [fill=xdxdff] (5.838486280754627,0.) circle (2.5pt);
\draw [fill=uuuuuu] (0.,1.3952) circle (1.5pt);
\draw[color=uuuuuu] (0.2753316280626758,1.1632738299093763) node {$r_s$};
\end{scriptsize}
\end{tikzpicture}
\end{center}
We now show that the following:
\begin{proposition}
The trajectories of the Kepler problem in General relativity (\ref{KPGR}) are dually parallel to ($\alpha=\frac12$) elliptic spirals (\ref{ellspirals}).
\end{proposition}
\begin{proof}
Starting with pedal equation for trajectories (\ref{KPGRp}) and applying dual parallel transform $E_\gamma^\star$ we get
$$
\frac{1}{p^2}=d+\frac{r_s}{a^2 r}+\frac{r_s}{r^3}\qquad \stackrel{E^\star_\gamma}{\longrightarrow}\qquad \frac{1}{p^2}=\tilde a+\frac{\tilde b}{r}+\frac{\tilde d}{r^2}+\frac{r_s}{r^3},
$$
where
\begin{align*}
\tilde a&:=d-\gamma\frac{r_s}{a^2}-\gamma^2-\gamma^3 r_s,\\
\tilde b&:=\frac{r_s}{a^2}+2\gamma+3r_s \gamma^2,\\
\tilde d&:=-3\gamma r_s.\\
\end{align*}
Choosing $\gamma$ such that $\tilde a=0$, which is always possible since a third degree algebraic equation has always a real solution, we see that the equation becomes the equation for elliptic spirals (\ref{ellspirals}) for the case $\alpha=\frac12$.
\end{proof}
\begin{remark}
By the same argument it can be shown that any curve of the form
$$
\frac{1}{p^2}=a+\frac{b}{r}+\frac{c}{r^2}+\frac{d}{r^3}, \qquad d\not=0,
$$
is dually parallel to ($\alpha=\frac12$) elliptic spirals.

\end{remark}
\subsection{Dark Kepler problem}\label{DKP}
Consider the dynamical system:
$$
\ddot x=-\frac{M}{\abs{x}^3}x+\frac{F}{\abs{x}}x-\omega^2 x,\qquad F,M\geq 0,
$$
which generalizes Kepler problem to include in addition to the gravitational effect of a central body (the $M$ term), also that of a homogeneous spherical bulk of dark matter around this central body (the $\omega^2$ term) and the dark energy -- i.e. constant outward repulsing force (the $F$ term).

In the pedal coordinates this takes form by Theorem \ref{cfth}:
\begin{equation}\label{dkpzadani}
\frac{L^2}{p^2}=\frac{2M}{r}+2Fr-\omega^2 r^2+c.
\end{equation}
Passing to the rotating frame of reference with angular velocity $\frac{\omega}{L} $ (using $A_{\frac{\omega}{L}}$ (\ref{Aomega})) this transforms to
$$
\frac{\zav{L+\omega r^2}^2}{p^2}=\frac{2M}{r}+2Fr+c+2\omega L.
$$
Comparing with the pedal form of Cartesian oval $\abs{x}+\alpha\abs{x-a}=C$ (See Example \ref{Covalex}): 
$$
\frac{\zav{b^2-(1-\alpha^2)r^2 }^2}{4p^2}=\frac{Cb^2}{r}+(1-\alpha^2)C r -\zav{(1-\alpha^2)C^2+b^2},
$$
where $b^2:=C^2-\alpha^2\abs{a}^2$.

We can see that these equations match if there exist a constant $\mu$ such that:
\begin{align*}
2L&=b^2\mu\\
2\omega &=- (1-\alpha^2)\mu\\
2M &=Cb^2 \mu^2\\
2F &=(1-\alpha^2)C \mu^2\\
c+2\omega L&=-\zav{(1-\alpha^2)C^2+b^2}\mu^2.\\
\end{align*}
This is a system of 5 algebraic equation in 4 unknowns $(\alpha,C,b^2,\mu)$, which cannot be, in general, satisfied if all the equations are independent. Hence, there has to be some connection between coordinates $(M,F,L,\omega,c)$. This connection is 
$$
FL+\omega M=0.
$$
Actually, it is enough that
$$
F^2 L^2=\omega^2 M^2,
$$
since the sign of $\omega$ (direction of rotation) can be chosen arbitrarily.

With this constrain the solution is
\begin{align*}
b^2&=\frac{2L}{\mu},\\
\alpha^2&=1-\frac{2\omega}{\mu},\\
C&=\frac{M}{L\mu},\\
2L\mu^2+(c+2\omega L)\mu+2\omega \frac{M^2}{L^2}&=0.
\end{align*}
The discriminant of the last equation is (using $\omega M=-FL$)
$$
D=(c+2\omega L)^2+4^2 FM\geq 0, 
$$
hence the solution exists and is given by
$$
\mu=\frac{\sqrt{D}-(c+2\omega L)}{4L}.
$$
We can see that $\frac{\mu}{L}>0$ and since $\omega <0$ (by agreement) the first three equations can be solved as well.

This solution assumes that $L\not=0, F\not=0$. For the special case $F=\omega=0$ (the usual Kepler problem) one gets singular solution with $\alpha=\pm 1$, i.e. an ellipse or a hyperbola. The solution of the case $L=0$ is a line segment with the origin as one its endpoints, thus $p=0$ in pedal coordinates. This can be seen by multiplying the original equation (\ref{dkpzadani}) by $p^2$ and letting $L\to 0$. 

Cartesian oval offers therefore a specific solution to the dark Kepler problem (in a suitable rotating frame of reference).

Furthermore, straightforward calculations shows that Dark Kepler problem is invariant under the transforms $E^\star_{\alpha} B_\alpha $ (\ref{Balpha}), specifically:
$$
\frac{L^2}{p^2}=\frac{2M}{r}+2Fr-\omega^2 r^2 +c\qquad \stackrel{E^\star_\alpha B_\alpha}{\longrightarrow} \qquad \frac{\tilde L^2}{p^2}=\frac{2\tilde M}{r}+2\tilde Fr-\tilde\omega^2 r^2 +\tilde c,
$$ 
where 
\begin{align*}
\tilde \omega^2&:=\omega^2+2F\alpha-\alpha^2 c+ 2\alpha^3 M +\alpha^4 L^2,\\
\tilde F&:=F-\alpha c+3\alpha^2M+2\alpha^3 L^2 & &=\frac12 \partial_\alpha \tilde \omega^2,\\
\tilde c&:= c-6\alpha M-6\alpha^2L^2& &=-\frac12\partial^2_\alpha \tilde \omega^2,\\
\tilde M&:=M+2\alpha L^2& &=\frac{1}{12}\partial_\alpha^3 \tilde \omega^2,\\
\tilde L^2&:=L^2& &=\frac{1}{24}\partial_\alpha^4 \tilde \omega^2.
\end{align*}

This means that starting with parameters $(L^2,2M,2F,\omega^2,c)$ it might be possible to transform them into parameters $(L^2,\tilde M,\tilde F,\tilde \omega^2, \tilde c)$ for which it holds
$$
\tilde F^2 L^2=\tilde\omega^2 \tilde M^2.
$$

Remarkably, this can \emph{always} be done. The equality above is an algebraic equation in $\alpha$ of the 6th order but (fortunately) all the coefficients of $\alpha^6,\alpha^5,\alpha^4$ are zero. Hence we are left with the equation of the third order:
$$
\left( -4\,F{L}^{4}-2\,cM{L}^{2}-2\,{M}^{3} \right) {\alpha}^{3}+
\left( -4\,{\omega}^{2}{L}^{4}+{c}^{2}{L}^{2}-2\,FM{L}^{2}+c{M}^{2}
\right) {\alpha}^{2}
$$
$$
+ \left( -4\,{\omega}^{2}{L}^{2}M-2\,cF{L}^{2}-2\,F{M}
^{2} \right) \alpha+ \left( FL-\omega M \right) \left( FL+\omega M \right)=0 ,
$$
which has always a real solution unless the leading term $2\,F{L}^{4}+\,cM{L}^{2}+\,{M}^{3}$ vanishes.

If the leading term vanishes the equation becomes
$$
(2\alpha L^2+M)^2 (F^2L^2-\omega^2 M^2)=0,
$$
which has evidently a real solution as well.

\end{document}